\documentclass[nofootinbib, aps, pra, 11pt, notitlepage,reprint,twocolumn,superscriptaddress]{revtex4-2}

\usepackage{our_style}

\usepackage{import}

\date{\today}

\begin{abstract}
Readout errors contribute significantly to the overall noise affecting present-day quantum computers. However, the complete characterization of generic readout noise is infeasible for devices consisting of a large number of qubits. Here we introduce an appropriately tailored quantum detector tomography protocol, the so called Quantum Detector Overlapping Tomography, which enables efficient characterization of $k-$local cross-talk effects in the readout noise as the sample complexity of the protocol scales logarithmically with the total number of qubits. We show that QDOT data provides information about suitably defined reduced POVM operators, correlations and coherences in the readout noise, as well as allows to reconstruct the correlated clusters and neighbours readout noise model. Benchmarks are introduced to verify utility and accuracy of the reconstructed model. We apply our method to investigate cross-talk effects on 79 qubit Rigetti and 127 qubit IBM devices. We discuss their readout noise characteristics, and demonstrate effectiveness of our approach by showing  superior performance of  correlated clusters and neighbours over models without cross-talk in model-based readout error mitigation applied to energy estimation of MAX-2-SAT Hamiltonians, with the improvement on the order of 20\% for both devices. 

\end{abstract}

\begin{document}

\title{Efficient reconstruction, benchmarking and validation of cross-talk models\\ in readout noise in near-term quantum devices 
}
\author{Jan Tuziemski}
\affiliation{\cft}
\author{Filip B. Maciejewski}
\affiliation{\cft}		
\affiliation{Research Institute for Advanced Computer Science (RIACS), USRA, Moffett Field, CA, USA}
\author{Joanna Majsak}
\affiliation{\cft}
\affiliation{Quantum Research Centre, Technology Innovation Institute, Abu Dhabi, UAE}
\affiliation{Faculty of Physics, University of Warsaw, Pasteura 5, 02-093 Warsaw, Poland}
\author{Oskar S{\l}owik}
\affiliation{\cft}	 
\author{Marcin Kotowski}
\affiliation{\cft}	
\author{Katarzyna  Kowalczyk-Murynka
}
\affiliation{\cft}	
\author{Piotr Podziemski}
\affiliation{\cft}	

\author{Micha\l\ Oszmaniec}
\affiliation{\cft}
\affiliation{NASK National Research Institute, Kolska 12,  01-045 Warsaw,
            Poland}
\email{oszmaniec@cft.edu.pl}

\maketitle

\section{Introduction}

Quantum computing shows considerable promise of outperforming classical computers in certain tasks, ranging from number factoring to quantum chemistry problems
\cite{ShorAlgorithm,BrandaoSDP,Harrow_2009,QuantumChemistrySpeedup, Pirnay2023superpolynomial}. However, present-day quantum computers are affected by noise, and execution of fault tolerant quantum algorithms is impossible.  It is still discussed whether NISQ devices can provide a quantum advantage \cite{Arute2019quantum,Pan2022SolvingSycamore,JianWeiPanQuantumAdvantage,Hsin-YuanAdvantage,farhi2014quantum,Cerezo2021VQAs,StilckLimitationsOptimization, DePalmaLimitationsVariational, AharonovGaoPolynomialTime, Oh2023classical}.

 Since gate errors limit depths of circuits that can be reliably performed, their efficient characterization \cite{EmersonSymetric, MagesanRandomizedBenchmarking, CrosstalkInQuantumProcessorsBlume-Kohout,  Harper2020efficient, HashimRandomizedCompiling,GateSetTomography_Nielsen2021,CompresiveGST_Brieger,Eisert2020Review,HelsenGeneralFramework,Helsen2023ShadowEstimation} and mitigation has attracted a lot of attention in recent years  \cite{Wallman2016noise, Temme2017error, Kandala2019error, huggins2020efficient, Endo2018, BergSparsePauli, sun2020mitigating, CaiMitigationReview,Quek2023exponentially}. Besides gate imperfections a significant source of errors in NISQ devices is readout of qubits, and a dedicated set of characterization tools \cite{Chen2019detector,Maciejewski2021modeling,Geller2020efficient,lilly2020modeling,Hamilton2020scalable} and error mitigation methods has been developed \cite{NationMitigation,seo_mitigation_2021,Kwon2020hybrid,Nachman2019unfolding,Bravyi2020mitigating,zheng2020bayesian,WangVonNeumann,funcke2020measurement,Maciejewski2020mitigation}. The recent data for superconducting architectures \cite{GoogleQEC, IBMResources,Arute2019quantum, GoogleQECexperiment}  show that single qubit readout errors are of the order of few percent and thus comparable to two-qubit gate errors, and this does not take into account cross-talk effects. Needless to say, measurements are indispensable part of virtually all quantum algorithms.

Mitigation of readout errors allows to improve estimation of expectation values of observables, a task ubiquitous in quantum information protocols such as variatioal algorithms, and classical shadows protocols \cite{Chen2021RobustShadowEstimation,Koh2022classicalshadows,McNMO2022}. Several error mitigation techniques have been developed \cite{NationMitigation,WangVonNeumann,Bravyi2020mitigating,Maciejewski2021modeling,Geller2020efficient,Nachman2019unfolding,Geller2020rigorous,zheng2020bayesian,funcke2020measurement,van_den_Berg_ModelFreeMitigation, Geller2021ConditionallyRigorous,seo_mitigation_2021,Smith:2021iwt,Hicks2022ActiveREM,BoREMDSparseResults,Kim2021REMDeepLearning}.
However, efficient readout noise modeling is a challenging task since complexity of general noise models grows exponentially with the system size, and characterization based on direct quantum detector tomography \cite{Lundeen2008} is feasible only for a few qubit systems \cite{Chen2019detector,Maciejewski2020mitigation}. To circumvent this difficulty several approaches have been developed. One example is Tensor Product Noise Model (TPM), in which the total noise matrix is approximated as a tensor product of single qubit noise matrices  \cite{Maciejewski2020mitigation,Bravyi2020mitigating,Smith:2021iwt}, at the expense of not taking cross-talk into account. Other approaches to cross-talk readout noise modeling based on assumptions on the noise structure, such as a model based on Continuous Time Markov Processes \cite{Bravyi2020mitigating} or a heuristic model presented in \cite{Geller2020efficient,Geller2021ConditionallyRigorous} have not been applied to model readout-noise on large scale quantum devices. 
 
 In this work we present an efficient method of $k$-local cross-talk readout noise characterization. This method is based on Quantum Detector Overlapping Tomography (QDOT) protocol, which is  schematically presented in Fig. \ref{fig:DOT_scheme}. This single-shot protocol characterizes reduced POVMs operators, which correspond to a particular marginalization of a $N$ qubit POVM: The qubits are divided into a $k$ qubit subset, and its $N-k$ qubit complement. The input state on the subset is an arbitrary state, and the state on the complement is fixed. After the measurement is performed, the results on the subset are saved and on the complement forgotten. QDOT protocol is an extension of Overalpping Tomography protocol for quantum states \cite{Cotler2020quantum} to the case of quantum measurements \cite{Maciejewski2020mitigation}. QDOT fits into the broader framework of randomized measurement techniques \cite{Elben_RandomizedToolbox}, such as classical shadows \cite{Huang2020predicting}, and classical shadows of quantum processes \cite{LevyShadowProcessTomography, KunjummenShadowProcessTomography}. We prove that the sample complexity of the protocol scales exponentially in $k$ and logarithmically in $N$.  The reduced POVM operators are used to analyze $k$-local cross-talk effects in readout noise, and  to reconstruct the correlated clusters noise model (discussed e.g. in \cite{Maciejewski2021modeling}).  We apply the method to characterize $k$-local cross-talk in readout noise of large scale NISQ devices - 79 qubits Rigetti Aspen-M-3, as well as 127 qubit IBM Cusco.  To demonstrate potential of the method for error mitigation we perform benchmarks based on energy estimation of MAX-2-SAT Hamiltonians ground states, and observe that error mitigation based on the reconstructed correlated clusters noise model leads to a significant improvement of the results. Design and implementation of the experiments, as well as data analysis was obtained using QREM package \cite{qrem}.

We start by reviewing the main concepts used throughout the work.

\onecolumngrid

\begin{figure*}
\centering
\includegraphics[width=0.99\textwidth]{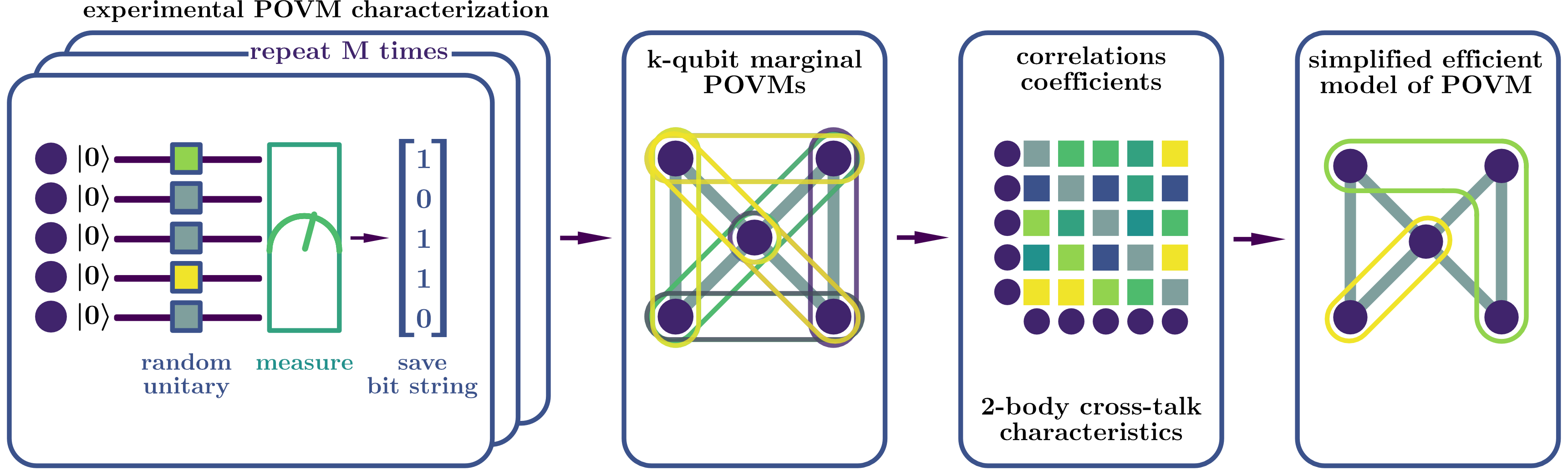}
\caption{\label{fig:DOT_scheme}
Schematic representation of the noise characterization procedure proposed in this work. In the first step quantum circuits for Quantum Detector Overlapping Tomography are designed and implemented on a device. In the post-processing stage QDOT data are used to reconstruct all reduced POVM operators acting on $k$ qubits. The reduced POVM operators are further used to provide various characteristics of cross-talk effects in the readout error.  Finally, a CN noise model best describing the experimental data is found. 
}
\label{fig:schematic_strategy}
\end{figure*}

\begin{figure*}
\centering
\includegraphics[height=4cm]{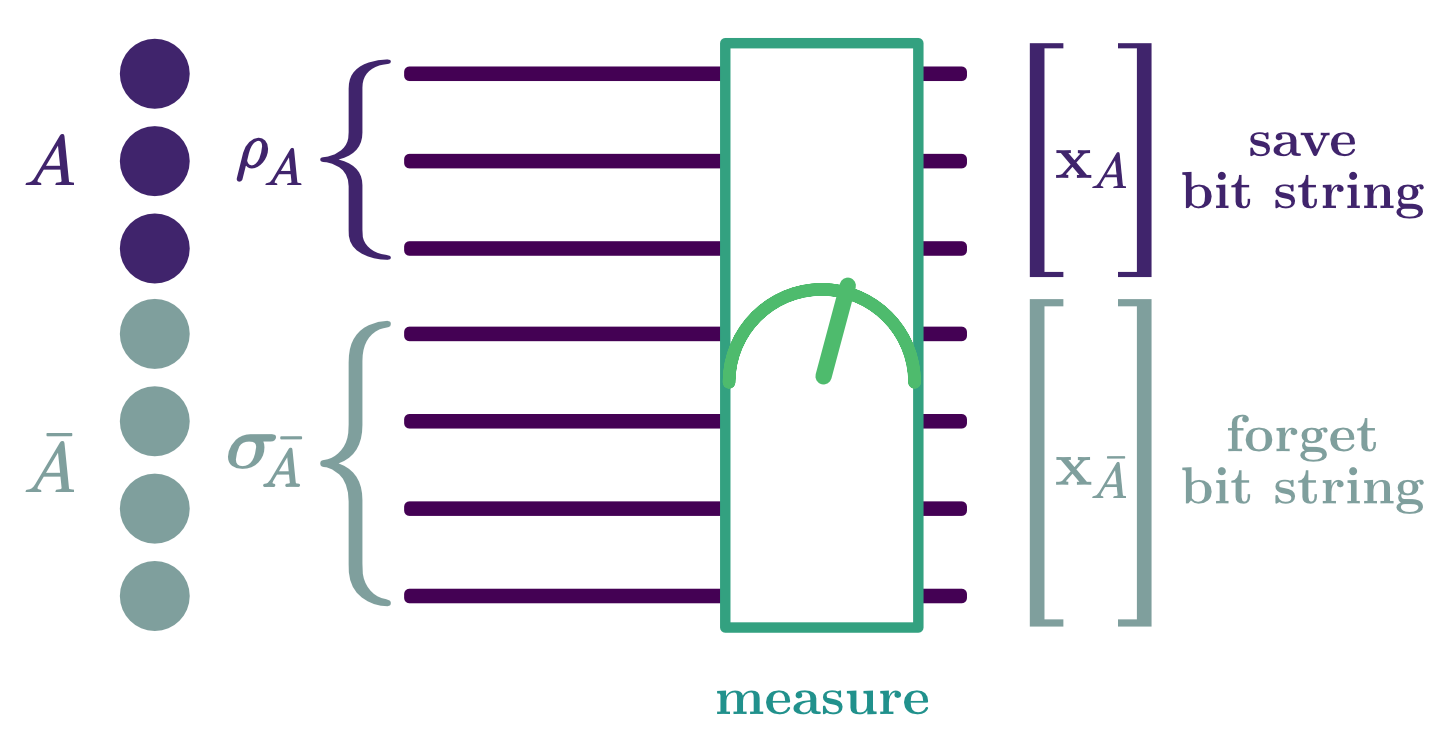}
\caption{\label{fig:marginal_measurement}
Schematic representation of the reduced measurement process defined in Eq. (\ref{eq:reduced_povm_definition}). The measurement of an input state on subsystem $A$ - $\rho_{A}$ is performed with a fixed state on $\bar{A}$ (the complement of $A$) - $\sigma_{\bar{A}}$. Measurement outcomes on subsystem $A$ are registered whereas those on $\bar{A}$ are forgotten. 
}
\end{figure*}

\clearpage

\twocolumngrid

\subsubsection{Basic definitions}

We are interested in errors occurring during a measurement process, which in quantum mechanics is described by a Positive Operator-Valued Measure (POVM).
A POVM $\M$  with $m$ outcomes performed on a $d$ dimensional Hilbert space consists of $m$ positive operators, so-called effects, that sum up to identity  
 \begin{eqnarray}
   \M = \left\{ M_\x \right\}_{\x=1}^m \; \; \; \; \forall_\x M_\x \geq 0 \; \; \; \;  \sum_{\x=1}^m M_\x = \mathbb{I}.     
\end{eqnarray}

The probability of obtaining outcome $\x$ while performing a measurement $\M$ on a quantum state $\rho$ is given by the Born rule: $\text{Pr}\left( \x \vert \M,\rho \right) =  \tr \left(M_\x \rho  \right)$ . 
As a model for an ideal detector, we will consider a $d=2^N$-outcome projective measurement of $N$ qubits $\P=\cbracket{\ketbra{\x_i}{\x_i}}_{i=0}^{2^N-1}$, with $\x_i$ is a computational basis state, which is typically a target measurement in $N$ qubit quantum devices.

\subsubsection{Readout noise modeling}\label{sec:background_modeling}

Let us now consider a POVM $\M$, which is a reconstruction of an ideal computational basis projective measurement $\P$. Due to the readout noise $\M$ may no longer be a projective measurement in computational basis: Effects of $\M$ may possess off-diagonal elements with respect to the computational basis of $\P$. In the multi-qubit case effects of $\M$ may not be decomposable as tensor products of single effects elements. Such properties of $\M$ constitute genuine quantum noise effects, and in cases when they are significant, the noise transforming $\P$ to $\M$ cannot be accurately described by a classical noise model. However, several investigations (see, e.g., \cite{Maciejewski2020mitigation,Chen2019detector,Geller2020efficient,Maciejewski2021modeling}), have shown that, to a good approximation, readout noise can be modeled as a classical stochastic process. In such a case POVM $\M$ is related to $\P$ via a stochastic transformation $\Lambda$: $\M = \Lambda \P$ so that $M_\x = \sum_{\x'} \Lambda_{\x \x'}P_{\x'}$. Due to the linearity of Born rule it immediately follows that a similar relation holds for probabilities describing statistics of the ideal $\p^{ideal}$  and the noisy $\p^{noisy}$  device: $\p^{noisy}= \Lambda \p^{ideal}$ \cite{Maciejewski2020mitigation}. The action of stochastic classical noise can also be interpreted as a reshuffling of results that is done in the post-processing stage.

\subsubsection{Correlated readout noise model}

In this work we aim to characterize the cross-talk noise model that was considered in e.g. \cite{Maciejewski2021modeling}. In this model qubits exhibiting strong correlations in the readout noise form clusters $\cluster_{\cindex}$. Moreover, the noise affecting qubits in a cluster $\cluster_{\cindex}$  may depend on the state of a group of other qubits. If for a given cluster $\cluster_{\cindex}$  such a group existsit is called a neighborhood of $\cluster_{\cindex}$, and denoted as  $\mathcal{N}_{\cindex}$. The general noise model matrix $\Lambda$ is composed of the following matrix elements     
\begin{align}\label{eq:noise_model_correlated}
	\Lambda_{\x|\y} = \prod_{\chi} \Lambda^{ \y_{\mathcal{N}_{\chi}}}_{\x_{\cluster_\chi}|\y_{\cluster_\chi}} \ .
\end{align}
\noindent 
 where $\cbracket{\cluster_{\cindex}}_{\cindex}$  is a collection of clusters, which fulfill $\cluster_{\cindex} \cap \cluster_{\cindex'} = \emptyset$ if $\cindex \neq \cindex'$ and $\cup_{\cindex}\cluster_{\cindex} = \sbracket{N}$, and $N$ is the total number of qubits. $\Lambda^{ \y_{\mathcal{N}_{\chi}}}$ is a matrix describing measurement noise affecting qubits in cluster $\cluster_{\cindex}$, and it depends on the pre-measurement state $\y_{\mathcal{N}_{\chi}}$ of the qubits in the neighborhood $\mathcal{N}_{\cindex}$ . The matrices $\y_{\mathcal{N}_{\chi}}$ are left stochastic. The total number of parameters needed to specify the noise model of the form of Eq.  (\ref{eq:noise_model_correlated}) is $\sum_{\cindex}2^{|\cluster_{\cindex}|} 2^{|\mathcal{N}_{\cindex}|}$, where $|\cluster_{\cindex}|$ and $|\mathcal{N}_{\cindex}|$ is the size of $\cindex$'th cluster and its neighborhood. This in turn implies that, provided that clusters' and their neighborhoods' sizes are bounded by a constant, the description of the measurement noise model is efficient. The locality structure of the noise model needs to be established based on experimental data from the device. It is not necessarily influenced by the device architecture.

 Note that without the neighborhoods the noise model matrix can be decomposed as a tensor product over noise matrices of individual clusters

 \begin{align}\label{eq:noise_model_correlated_no_neighbors}
	    \Lambda = \bigotimes_\cindex \Lambda_{\cluster_{\cindex}}.
	\end{align}

As discussed in \cite{Maciejewski2021modeling}, according to the noise model defined in Eq. (\ref{eq:noise_model_correlated}) the relation between a noisy and ideal marginal probability distribution $\p^{noisy}_\S$, $\p^{ideal}_\S$ on $\S$, is given by the local noise matrices acting on the qubits from $\S$ and the joint neighborhood of $\S$, $\mathcal{N}\rbracket{\S} \coloneqq  \cup_{\cindex\in\mathcal{A}}\mathcal{N}_{\cindex} \setminus \S$ (the qubits which are neighbors of qubits from $\S$ but are not in $\S$ themselves). Due to this fact the in error mitigation one can not directly use inverse of the noise matrix inverse. Following \cite{Maciejewski2021modeling} we use an effective noise model on the marginal $\S$, i.e. a noise matrix averaged over states of qubits from the joint nieghborhood defined above  
\begin{align}\label{eq:marginal_lambda_average}
 \Lambda^{\S}_{av}\coloneqq \frac{1}{2^{|\mathcal{N}\rbracket{\S}|}} \sum_{\Y_{\mathcal{N}\rbracket{\S}}} \Lambda^{\Y_{{\mathcal{N}(\S)}}} \ .
\end{align}

\section{Results}
\subsection{Reduced quantum measurement operators}
\label{sec:Tomography}

For quantum states the definition of a state marginal is unambiguous. It is not the case for the measurement operators, where a reduction of a measurement operator to a subset of qubits may depend on a state measured on the subset's complement. In order to avoid this ambiguity, to define a reduced POVM we fix a state on the complement, and in most cases it is set to the maximally mixed state, see Fig. \ref{fig:marginal_measurement}.  

We denote the subset of $k$ qubits as $A$ and as $\bar{A}$ its $N-k$ qubits complement, where $N$ is the total number of qubits.
The reduced  POVM effect $ M^{A, \sigma_{\bar{A}}}_{\x_A}$ is of the form  
\begin{eqnarray}\label{eq:reduced_povm_definition}
M^{A, \sigma_{\bar{A}}}_{\x_A} \equiv \sum_{\x_{\bar{A}}} \tr_{\bar{A}} \left[ M_{\x_{A}\x_{\bar{A}}}  \left( I_A \otimes \sigma_{\bar{A}} \right) \right],
\end{eqnarray}
where $\sigma_{\bar{A}}$ is a fixed state on the complement, see also Fig. \ref{fig:marginal_measurement}, and $M_{\x_{A}\x_{\bar{A}}}$ is an effect of the total POVM $\M$ corresponding to the result $\x=\x_{A}\x_{\bar{A}}$, where $\x_{A}, \x_{\bar{A}}$ are outcomes on $A$ and $\bar{A}$ respectively.
The reduction considered here has an operational meaning: Reduced POVM elements correspond to a measurement process, in which subsystem $A$ is prepared in the state $\rho_{A}$, the state of the complement $\bar{A}$ is $\sigma_{\bar{A}}$, and the measurement results on subsystem $\bar{A}$ are forgotten (by summing over $x_{\bar{A}}$ ). 
Unless stated otherwise we will choose the state of the complement as $\sigma_{\bar{A}}=\tau_{\bar{A}} \equiv \frac{I_{\bar{A}}}{2^{N-k}}$, i.e. the maximally mixed state,  and denote 
\begin{eqnarray}
\label{eq:reducedPOVMIdentity}
M^{A}_{\x_A} \equiv M^{A, \tau_{\bar{A}} }_{\x_A} = \frac{1}{2^{N-k}} \sum_{\x_{\bar{A}}} \tr_{\bar{A}} \left( M_{\x_A \x_{\bar{A}} } \right).
\end{eqnarray}



In the next subsection we will present a protocol that allows for efficient estimation of a reduced POVM $\M^A$, as well as its matrix elements  $ M_{\x_A,\y_A}^{A} \equiv  \text{tr} \left(   M^A_{\x_A}  |\y_A \rangle \langle \y_A| \right)$, where the comma is used to separate measurement results $\x_A$ from the input state $\y_A$. In Appendix \ref{sec:channel_reductions} we show that the above reductions can be seen as special cases of general reductions of  quantum channels framework. For other approaches to channels reduction see works on  channel marginal problem  \cite{HsiehChannelMarginalProblem2021}, as well as on process shadow tomography \cite{LevyShadowProcessTomography,KunjummenShadowProcessTomography}.

\subsection{Quantification of quantum and classical effects in readout noise}
\label{sec:correlations_quantification}
In this section we utilize the reduced POVMs to define some basic indicators of quantum and classical effects in the readout noise. To this end we introduce the following theoretical quantities, which later are used in assessment of experimental data from quantum hardware.      

\subsubsection{Quantum correlations coefficients}

Pairwise correlations in the readout process are quantified using two-qubit quantum correlation coefficients, which are a generalization of coefficients introduced in \cite{Maciejewski2020mitigation}.
This construction is based on two types of reduced POVMs discussed in Section~\ref{sec:Tomography}. A 2-qubit reduced POVM  $\M^{\left\{ij \right\}}$, which is defined via taking the reduction with the maximally mixed state on the complement of qubits $i,j$ (cf. Eq. (\ref{eq:reducedPOVMIdentity})), is further reduced to a one-qubit POVM $\M^{i,\rho_j}$ , acting on qubit $i$, by fixing an input state $\rho_j$ on qubit $j$.

The quanum correlation coefficient $c^{\dist}_{j \rightarrow i}$ corresponding to distance $\text{dist}(\cdot,\cdot)$ is defined as 
\begin{eqnarray}
\label{eq:correlation_coefficient_general}
    c^{\dist}_{j \rightarrow i} = \sup_{\rho_j, \sigma_j} \dist \left( \M^{i,\rho_j}, \M^{i,\sigma_j} \right),
\end{eqnarray}
where the supremum is taken over all possible pairs of qubit states $j$, the reduced POVM effects $M^{i,\rho_j}_{\x_i}$ are defined according to Eq. (\ref{eq:reduced_povm_definition}), and $\dist(\cdot,\cdot)$ denotes either worst-case (it will be abbreviated as WC) or average-case distance (it will be abbreviated as AC) \cite{MaciejewskiAverageCaseDistance,MaciejewskiAverageCaseDistanceProofs}. For completeness definitions and basic properties of the distances are stated in Sec. \ref{sec:background_distances}.
The quantum correlation coefficient $c_{j\rightarrow i}$ quantifies how much the measurement statistics on qubit $i$ can differ depending on the input state on qubit $j$.
We call the correlation coefficients classical when the supremum in Eq. (\ref{eq:correlation_coefficient_general}) is restricted to pure states from computational basis. 
Hence for a fixed POVM the difference between quantum and calssical correlation coefficient is an indicator of genuine quantum effects in the readout errors.

Computation of quantum correlation coefficients and their properties are discussed in Appendix~\ref{sec:app:correlation_coefficients}, where we e.g. provide analytical form of the correlation coefficients for diagonal POVMs, and we discuss a heuristic optimization procedure in the quantum case.

\subsubsection{Coherent errors}

A distance between a reconstructed POVM $\M$ and its diagonal counterpart is a natural quantifier of coherent errors, as the classical noise does not introduce coherences in POVM's effects. Therefore, following \cite{Maciejewski2020mitigation},   we define coherence strength of a noisy $d$ dimensional POVM $\M$ associated with a measure of distance $\dist(\cdot,\cdot)$ as    
\begin{align}\label{eq::coherence_definition}
    \mathrm{CS}(\M)\left(\M, d\right) = \dist\left(\M, \Phi_{\text{dep}}(\M)\right) \ ,
\end{align}
where $\Phi_{\text{dep}}$ is the maximally dephasing channel that acts on POVM effect $M_\x$ as $\Phi_{\text{dep}}(M_\x) = \text{diag}(M_\x)$, with $\text{diag}(M_\x)$ denoting the diagonal part of $\M_\x$. The final measure  depends on the choice of distance $\dist(\cdot,\cdot)$. Here, in addition to the operational distance considered in \cite{Maciejewski2020mitigation}, we use also the average-case distance \cite{MaciejewskiAverageCaseDistance,MaciejewskiAverageCaseDistanceProofs}. For definition of the distances see Sec. \ref{sec:background_distances}.

This measure cannot be directly efficiently computed for generic high-dimensional POVMs.
However, the tomographic reconstruction of a POVM is not needed to lower bound the coherence strength. More precisely, in Sec. \ref{sec:app_cs_bound} we prove the following proposition  
\begin{prop}
\label{prop:cs_lower_bound}
For a POVM $\M$ and quantum states $|p\rangle, |q\rangle$ constructed from tensor product of Pauli $X,Y$ eigenstates the following bound holds 
\begin{eqnarray}
\label{eq:cs_lower_bound}
    \mathrm{CS}(\M)^{AC} \geq \frac{\tvd \left( \Pr(p), \Pr(q)   \right)}{d \| \delta_{p,q} \|_{HS}},
\end{eqnarray}
where $\Pr(p)$ ($\Pr(q)$) is probability distribution obtained by measuring $|p\rangle$ ($|q\rangle$)  with $\M$, and $\delta_{p,q} =  \left( | p \rangle \langle p | - | q \rangle \langle q | \right ) $, and $d$ is  dimension of the corresponding Hilbert space. 
\end{prop}
Note that estimation of the bound requires probability distributions steaming from quantum circuits preparing eigenstates of Pauli $X$ or $Y$ operators on individual qubit. Therefore, estimation of the bound involves less quantum resources than direct estimation of coherence strength.

Coherence strength of reduced few-qubits POVMs can be used to gain insights into properties of the global POVM. Large coherence strength of reduced POVMs is an indicator of low accuracy of a classical readout noise model for the considered device. Under the the assumption that that the global POVM $\M$ can be decomposed as a tensor product of reduced POVMs $\M = \bigotimes_i \M^{A_i}$, the coherence strength of reduced POVMs can be also used to upper bound coherence strength of the global POVM 
\begin{equation}
    \label{}
  \mathrm{CS}\left(\M, d\right)  \leq \sum_i \mathrm{CS}\left(\M^{A_i}, d_i\right).
\end{equation}
The above inequality is derived using subadditivity w.r.t. tensor product of distances involved in the definition of the coherence strength.

\subsection{Rigorous  overlapping tomography of quantum measurements}

Inspired by Overlapping Tomography (OT) protocol (see \cite{Cotler2020quantum}), we use a tailored single-shot tomography procedure, Quantum Detector Overlapping Tomography to perform $k$-qubit detector tomography, and reconstruct $k$-qubit reduced POVMs in a parallel way.  Our main result is that sample complexity, i.e. the number of circuits, of a QDOT experiment scales as $O(6^k\log(N))$ -  exponentially in $k$, i.e. considered locality of readout noise, and logarithmicaly in $N$. As a result, as long as locality of noise is constrained, characterization of readout noise is efficient for large number of qubits. We consider also a variant of QDOT protocol, in which only diagonal elements of reduced POVMs effects are reconstructed - Diagonal Detector Overlapping Tomography (DDOT), which can be useful e.g. when only stochastic features of noise are considered.

The full QDOT  protocol is presented in Appendix \ref{app:DOT}, here we sketch its main steps.  Firstly, a collection of random quantum circuits is generated, which consist of a single layer of one-qubit gates. The gates are chosen at random from a fixed set of one-qubit gates that prepare  eigenstates of Pauli operators or computational basis states in the case of DDOT (the protocol can use any informatially complete basis, the only difference is then the input state on $\bar{A}$, and hence the reduced POVMs will be of a form of Eq. (\ref{eq:reduced_povm_definition})). The errors introduced by the gates are neglected. Single copy measurements of the states from the collection are performed, and statistics of outcomes are used to construct an unbiased estimator of the reduced POVMs effects Eq. (\ref{eq:reducedPOVMIdentity}). Details, and sample complexity analysis of the full protocol are presented in Appendix  \ref{app:DOT}).

We provide results regarding sample complexity of the estimation of  noisy POVM $\M^{A}$ matrix elements. 
 \\

\textbf{Theorem 1 } \textit{Experimental QDOT data allow to estimate matrix elements of a noisy POVM $\M^{A}$ within $\epsilon$ error for all possible choices of the $k$ qubit subset $A$, provided that the number $N_{Mes}$ of circuits used in the experiment scales as $O\left(\frac{k 6^k \log (N)}{\epsilon^2} \right)$. Diagonal entries of $\M^{A}$ can be estimated by means of DDOT protocol, and the scaling is $O\left(\frac{k 2^k \log (N)}{\epsilon^2} \right)$   } \\
 
We also bound sample complexity needed to reconstruct a POVM $\M^{A}$ with a precision set by the average-case distance Eq. (\ref{eq:average_dist}) and operational distance Eq. (\ref{eq:operational_dist}). 
\textbf{Theorem 2 } \textit{Experimental QDOT data allow to estimate POVMs $\M^A$ for all possible choices of the $k$ qubit subset A and within $\epsilon$ error in the average case distance if the number $N_{Mes}$ of circuits used in the experiment scales as $O\left(\frac{k 6^k \sqrt{2^k+1}\log (N) }{\epsilon^2} \right)$   . The scaling for operational distance is $O\left(\frac{k 12^k \log (N) }{\epsilon^2} \right)$}  \\

The above theorem is based on sample complexity of Choi-Jamiolkowski state reconstruction using least squares approach to quantum process tomography \cite{SurawyStepneyQPT}. Exact scaling of sample complexity, including constant factors not included above, as well as proofs of both theorems are presented in Appendix \ref{app:DOT}.

In the QDOT protocol each random circuit (input state) is measured only once. This is an undesirable feature from the point of view of  the current generation of NISQ devices, for which it is more practical to measure a particular circuit many times. In Appendix \ref{app:DOT} we analyze modified version of the QDOT protocl, in which measurement settings are reused. Our analysis closely follows recent works on multi-shot classical shadows \cite{helsen2022thrifty, Zhou2023multishot}.

\clearpage

\clearpage

\onecolumngrid
\begin{figure*}[t]
\includegraphics[width=0.475\textwidth]{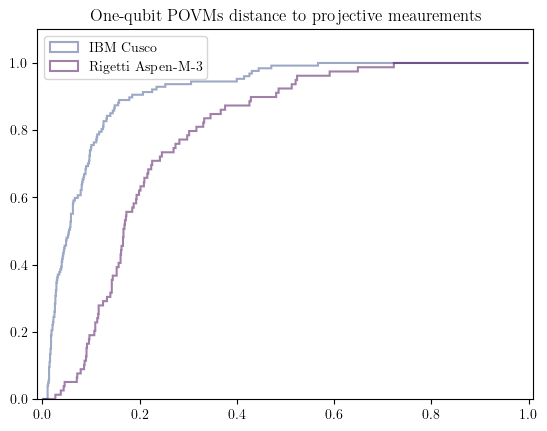}  
\includegraphics[width=0.475\textwidth]{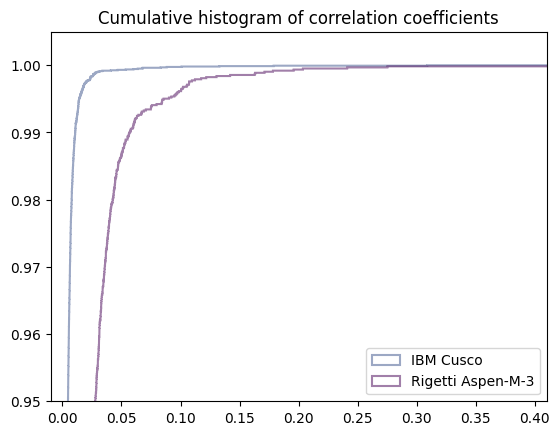} 
\caption{\label{fig:basic_properties} Left panel  -cumulative histogram of worst case distance distances between reconstructed one-qubit POVMs and ideal single-qubit computational basis POVMs from reconstructed single-qubit POVMS. The blue lines correspond to results obtained for IBM Cusco, and the purple to  Rigetti Aspen.    
Right panel - cumulative histograms for classical correlations coefficients calculated with worst-case distance. Note that majority of the coefficients is small -  the Y axis starts with $0.95$.
 }
\end{figure*}
\begin{figure*}
\centering
\includegraphics[width=0.435\textwidth]{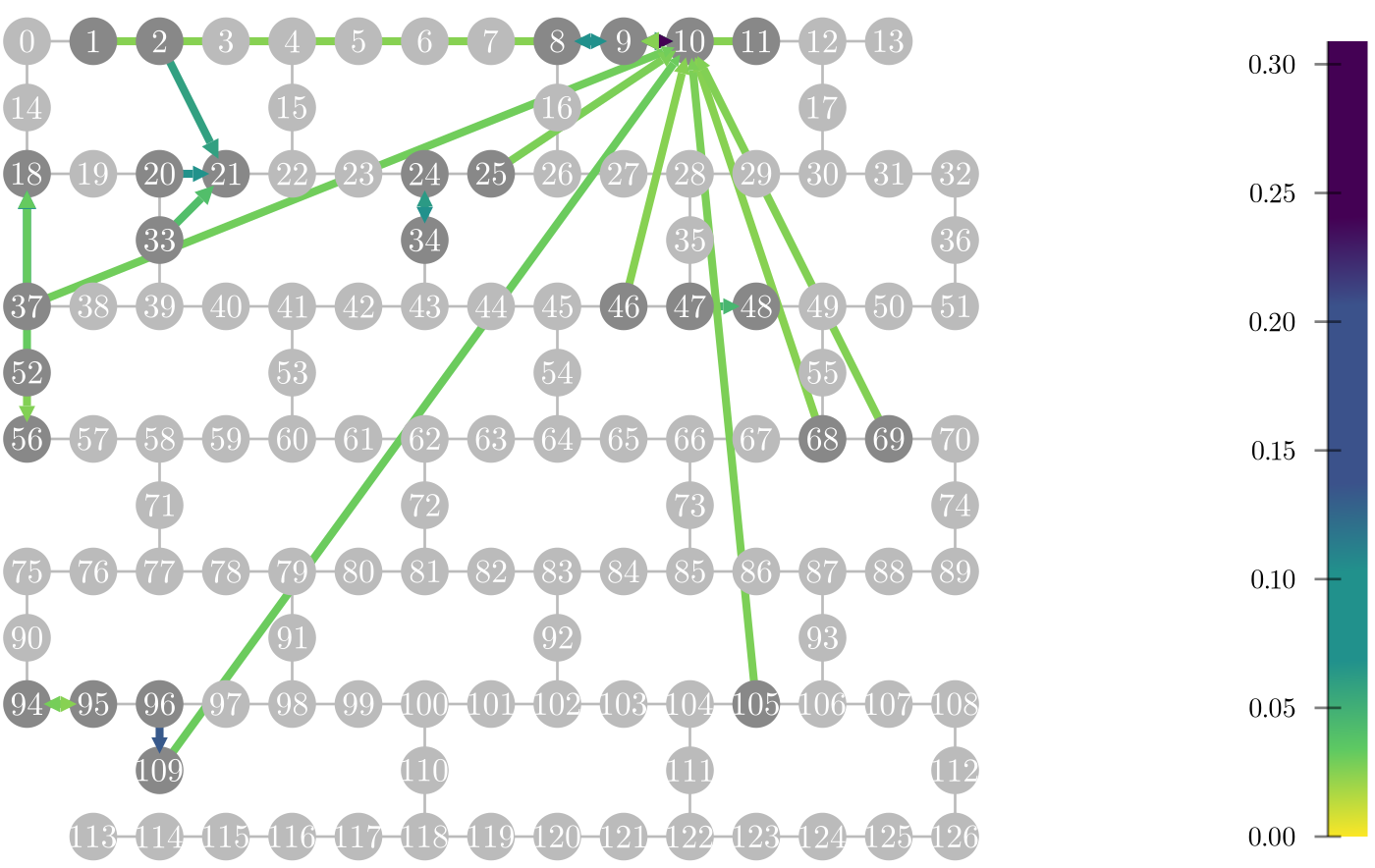} 
\includegraphics[width=0.5\textwidth]{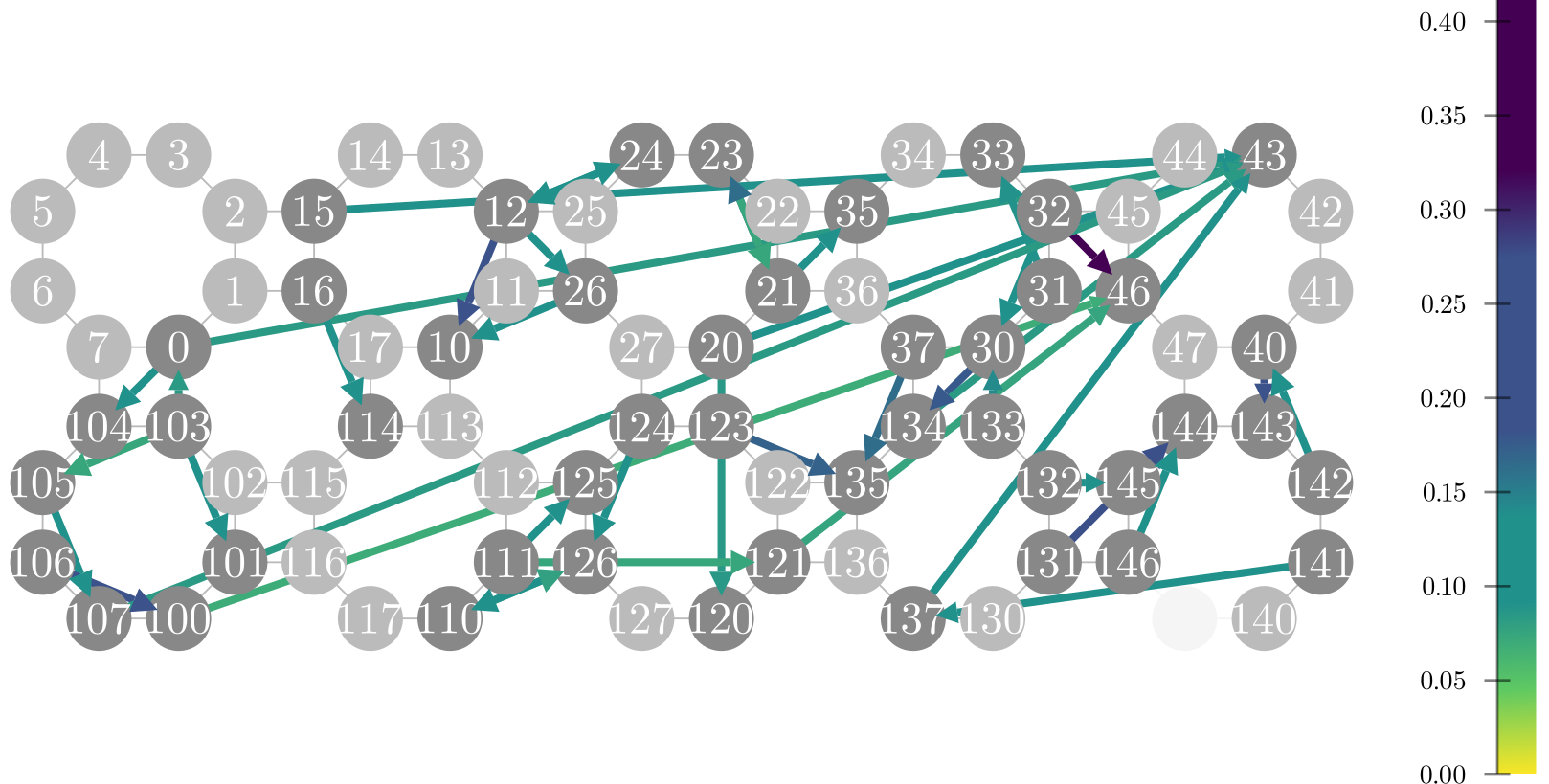}
\caption{\label{fig:connectivity_graphs} Correlation coefficients plotted on connectivity graph of IBM Cusco - left panel,
and  Rigetti Aspen-M-3 - right panel
The nodes and structure of the graph correspond to the topology of the device. 
The color of the arrows indicates correlations' strength, note different scale for both devices. Arrows between qubits indicate direction of the most profound influence. For IBM Cusco correlations higher than the threshold $ 3\%$ are shown, whereas for Rigetti Aspen-M-3 the threshold is $7\%$ 
}
\end{figure*}

\twocolumngrid 

\clearpage

\subsection{Reconstructing noise structure from detector overlapping tomography}
\label{sec:model_fitting}

Reconstruction of the correlated clusters noise model involves finding an appropriate cluster assignment. Here we describe an algorithm that returns such a cluster assignment based on the pairwise correlation coefficients  $c_{i \rightarrow j}$.


Finding cluster assignments is an optimization problem that involves a certain trade-off between the maximal acceptable size of clusters and the accuracy of the model: Accuracy requires assigning all correlated qubits to the same cluster, but sample complexity of required DOT experiments scales exponentially with the clusters' size. As a result, it is desirable to form clusters with limited number of qubits. Note that this is a different situation to that encountered e.g. in community detection problems, where in most cases one is interested in forming as large clusters as possible \cite{Fortunato2022community}.  The input data to the algorithms is the collection of the quantum correlation coefficients $c_{i \rightarrow j}$ defined in Eq. (\ref{eq:correlation_coefficient_general}), which are computed using reconstructed reduced POVMs, and  the maximal allowed clusters' size $C_{max}$. To quantify the trade-off we introduce a heuristic objective function $\phi(\mathcal{P})$, where  $\mathcal{P} = \cbracket{\cluster_{1},\dots, \cluster_{x}}$ denotes a cluster assignment. Positive contributions to this function come from a measure of a total correlation strength in a considered cluster assignment, which reflects the fact that strongly correlated qubits should belong to the same cluster to accurately describe cross-talk observed in the experimental data. To prevent formation of too large clusters, there is a negative contribution from a penalty function, which depends on a size of a cluster. Exact form of the function, as well as description of the optimization algorithm, can be found in the Appendix \ref{sec:det_clust}.

\subsection{Benchmarking correlated noise models }
\label{sec:benchmarking}

Accuracy and utility of the reconstructed noise model against experimental data is tested with two benchmarks. They are based on the energy estimation of ground states of classical local Hamiltonians (i.e. involving only Pauli $Z$ operators). These states provide a good reference as their energy can be efficiently estimated classically, and noise in their  experimental energy estimates is  dominated by readout errors (their preparation requires only a single-layer of one-qubit gates, which typically have a high fidelity). 

The first benchmark is based on energy prediction. The reconstructed noise model is used to calculate (predict) a noisy energy expectation value on a classical computer, and this quantity is compared to energy expectation value obtained from experimental results. Specifically, we calculate a normalized difference
    \begin{eqnarray}
    \label{eq:benchmark_prediction_error}
&&\Delta E_{\mathrm{PRED}}(H)=\frac{|E_{\mathrm{PRED}}(H)- E_{\mathrm{EST}}(H)|}{N}, 
\end{eqnarray}
where $E_{\mathrm{PRED}}(H)$ is the energy expected value of Hamiltonian $H$ predicted by the reconstructed noise model, and $E_{\mathrm{EST}}(H)$ is its estimator obtained from experimental data. The number of qubits in a device $N$ is included as a normalization factor to compare results obtained on devices consisting of different number of qubits.

In the second benchmark the reconstructed correlated clusters noise model is used to perform error mitigation, by a direct computation of an appropriate inverse of the noise matrix (cf. Eq. (\ref{eq:noise_model_correlated})). For a detailed discussion of those computations see \cite{Maciejewski2020mitigation}, where all issues of the approach, such as no guarantee that the inverted noise matrix is stochastic,  are addressed. In the benchmark a difference between error-mitigated energy expectation value and the ideal energy of the Hamiltonians - $\Delta E_{\mathrm{MIT}}(H)$ is compared to a difference between noisy energy expectation value (without error mitigation) and the ideal one $\Delta E_{\mathrm{EST}}(H)$. More specifically :
    \begin{eqnarray}
\label{eq:benchmark_mitigation_error}
&&\Delta E_{\mathrm{MIT}}(H)=\frac{| E_{\mathrm{MIT}}(H)- E_{\mathrm{TH}}(H)|}{N}\ , 
\end{eqnarray}
where $E_{\mathrm{TH}}(H)$ is the theoretical value of energy expectation value (estimated on a classical computer), and $\Delta E_{\mathrm{EST}}(H)$ is defined be replacing $E_{\mathrm{MIT}}(H)$ with the energy estimate  
$E_{\mathrm{EST}}(H)$ (i.e. the one without error mitigation) in the formula above.

\section{Experimental results}

Here we present experimental results concerning implementation of DDOT protocol, and subsequent characterization of the correlated readout noise, as well as benchmarks of the reconstructed CN noise model. The DDOT experiments were performed on IBM 127 qubits devices Cusco, Nazca and Brisbane, and on Riggeti Aspen-M-3 79 qubits device. The details of experiments are summarized in Sec.  \ref{app:experimental_details}. Since the IBM devices the results are qualitatively similar, here we discuss only those for IBM Cusco to keep the presentation concise. The experiments were designed and implemented, as well data analysis was performed using QREM Package, which provides a versatile set of tools for the characterization and mitigation of readout noise in NISQ devices \cite{qrem}. The code and experimental data allowing to reproduce the analysis presented in this section is available in the online repository  \cite{digital_apendix}.

\subsection{Experimental analysis of single and two qubit readout noise}
\label{sec:onetwoqubit}

The first stage in post-processing of DDOT data is the reconstruction of the reduced one- and two-qubit POVMs. To quantify their fidelity a distance to respective computational basis measurements is computed. The results for worst case distance are presented in the left panel of Fig. \ref{fig:basic_properties}.  The main conclusion is that there is a significant readout noise as a considerable fraction of reconstructed reduced POVMs deviates substantially from computational basis projective measurements.

Subsequently, the reduced two-qubit POVMs are used  to compute classical correlation coefficients (cf. Eq. ( \ref{eq:correlation_coefficient_general})) for all pair of qubits. A cumulative histogram of the correlation coefficients for both devices is presented in the right panel of Fig.~\ref{fig:basic_properties}. In the case of IBM Cusco majority of the correlation coefficients is smaller or equal than $1\%$, and the maximal value is approximately 30\% (for qubits 9 and 10). For Rigetti Aspen-M-3  majority of correlations coefficients is smaller than $5\%$, and the maximal value is approximately 41\% (for qubits 32 and 46). To provide information about spatial character of the correlations in the readout noise we present correlation coefficients drawn on the connectivity graphs of both devices (Fig.~\ref{fig:connectivity_graphs}).
The main conclusion is that the pairwise readout error correlations are not restricted to qubits that are close on a connectivity graph. 

We also investigated coherent strength by computing lower bound on coherence strength for two-qubit POVMs (Eq. (\ref{eq:cs_lower_bound})). The results are presented in Fig \ref{fig:cs_bound} . In the case of IBM Cusco there is no indication of coherent errors in two qubit POVMs, as the value of the bound is within the estimation error. The situation is different in the case of Rigetti Aspen-M-3, where for qubit parir (41,43) the lower bound is approximately $0.1$, and the estimation error is $\sim 0.03$. This allows us to conclude that coherent readout errors are present on Rigetti Aspen-M-3 device.

\subsection{Noise models reconstruction}

We implemented the algorithm described in Sec. \ref{sec:model_fitting} to reconstructed CN noise models form experimental data. For each device we reconstructed noise models corresponding to maximal cluster sizes $\RaggedRight C_{\mathrm{max}} \in \{2,3,4\}$. The reconstructed noise models were benchmarked to test their utility. Performance of CN noise models with neighbours Eq. (\ref{eq:noise_model_correlated}) was similar to those without neighbours Eq. (\ref{eq:noise_model_correlated_no_neighbors}), and we present the results only for the latter.     

\subsection{Noise model benchmarks}

We implemented 300 ground states of local Hamiltonians on IBM Cusco, and  50 on Riggeti M3. In all cases the 2-local Hamiltonians were used.  Here we focus on the energy mitigation benchmark. Results of the energy prediction benchmark are discussed in Appendix \ref{sec:add_experiment_details}.

In Fig. \ref{fig:error_mitigation} the results of the error mitigation benchmark are presented.  In each plot there are three curves: The blue one shows data for the CN noise model, the purple Tensor Product Noise (TPN) model, and the green one shows data without error mitigation. Both plots show that the value of the normalized energy expectation value error is much smaller for the error mitigated data. For IBM Cusco we observe that the median of $\Delta E_{\mathrm{MIT}}(H)$ is smaller by at least a factor of $10$, whereas for Rigetti Aspen-M-3 by at least $5.9$. This significant improvement in accuracy of energy expectation value estimation implies that reconstructed correlated clusters noise model captures correlations in the readout noise, and can be used as an input for error mitigation techniques based on knowledge of noise structure. To verify the influence of cross-talk on the results of the benchmark, we also reconstructed the TPN model, in which correlations in the readout noise are not taken into account. We found that error mitigation based on reconstructed correlated clusters noise model performs better than that based on TPN: in the case of IBM Cusco the median of $\Delta E_{\mathrm{MIT}}(H)$  is smaller by $23\%$, whereas for Rigetti Aspen-M-3 by $20\%$. This significant improvement implies that to fully utilize the potential of readout-errors mitigation techniques, it is crucial to characterize cross-talk effects, and it is not sufficient to focus on uncorrelated readout errors.

\clearpage

\onecolumngrid

\begin{figure*}
\includegraphics[width=0.48\textwidth]{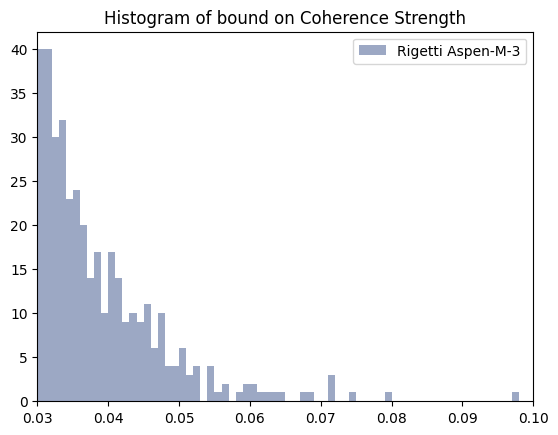} 
\caption{Histogram of coherence strength bound on two-qubit POVMs for Riggeti Aspen-M-3. The results below estimation error, which is on the order of $0.03$ are not shown. Details of estimation error are discussed in Sec. \ref{sec:app_cs_bound} For IBM Cuso all results are within the estimation error.}
\label{fig:cs_bound}
\end{figure*}

\begin{figure*}
   \includegraphics[width=0.48\textwidth]{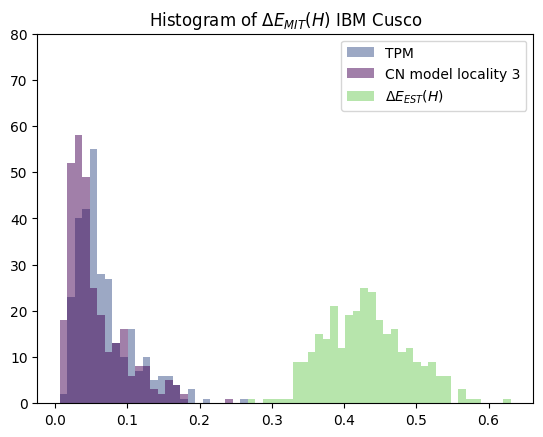}
   \includegraphics[width=0.48\textwidth]{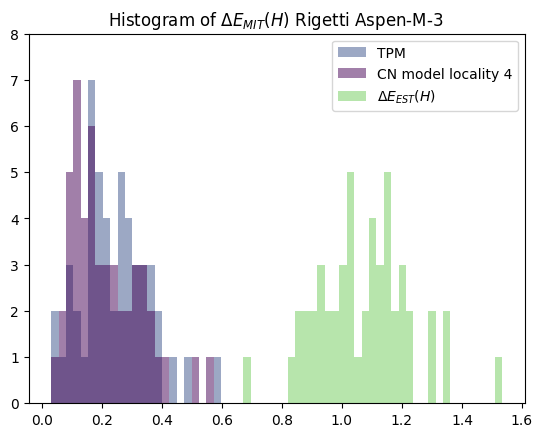}
\caption{\label{fig:error_mitigation} 
 Results of error-mitigation benchmark for \qibm-qubits on IBM Cusco - left panel, 
 and \qrig-qubits on Rigetti’s Aspen-M-3 - right panel. 
 In both cases a distribution of error in energy estimation for error mitigation based on product noise model (light purple), correlated clusters noise model (dark purple), as well as distribution of error in energy estimation for raw (unmitigated) data (light green) is presented. The error is calculated according to Eq. (\ref{eq:benchmark_mitigation_error}).
 The histograms were created from 300 Hamiltonians for IBM Cusco, and 50 Hamiltonians of the same locality Aspen-M3. In both cases Hamiltonians were built from single and two-body terms involving Pauli $Z$ operators. }
\end{figure*}

\twocolumngrid

\section{Discussion}

In this work we introduced a rigorous cross-talk readout noise characterization protocol based on Quantum Detector Overlapping Tomography procedure. The protocol is efficient as sample complexity of required quantum experiment scales logarithmicaly with the number of qubits, and can be applied to large-scale quantum devices. The quantum circuits implementing Quantum Detector Overlapping Tomography require only one-layer of one-qubit gates. The protocol provides various characteristics of the $k$-local cross-talk effects in the readout noise, and allows reconstruction of the CN readout noise model.  Therefore, the protocol has a potential to become a standard readout noise characterization subroutine for NISQ devices. We implemented the protocol on large-scale IBM and Rigetti devices, performed noise characterization, and reconstructed their respective CN noise models. For both devices we demonstrated also the utility of the CN noise models for readout error mitigation, and showed that error mitigation based on CN model performs better than the one based on the TP model,  with the improvement on the order of 20\%.  These results suggest that the increase in quantum resources needed for CN model reconstruction gives rise to a significant improvement of the readout errors mitigation efficiency, as compared to TP model. It is worth to mention the efficiency of the protocol implementation can be further optimized: Determination of clusters structure is based on correlation-coefficients, which are established using experimental data form QDOT protocol of locality 2. Therefore, it is possible to implement the CN model reconstruction protocol in a hybrid way: In the first step execute the QDOT protocol of locality 2, subsequently determine the cluster structure, then, if needed, perform additional tomography to estimate noise matrices associated with the reconstructed clusters. Such an approach would further reduce the number of quantum circuits used in the characterization protocol.

To the best of our knowledge this is the first readout cross-talk characteristic of large scale NISQ devices, which was made possible by development of theoretical tools presented here. In comparison previous studies such as e.g. \cite{Geller2020efficient,Maciejewski2021modeling} dealt with devices consisting of significantly smaller numbers of qubits.

This work opens further research directions. Most importantly, it would be interesting to apply similar reasoning to characterization of the cross-talk noise affecting quantum gates. In this case one can consider spatial and temporal corretaltions: Cross-talk in gates applied to different qubits, as well as a temporal (non-Markovain) errors that affect sequence of quantum gates applied to a fixed group of qubits.     

\begin{acknowledgments}
This work was supported by the TEAM-NET project co-financed by the EU within the Smart Growth Operational Program (contract no. POIR.04.04.00-00-17C1/18-00). 
We acknowledge the use of IBM Quantum services for this work. The views expressed are those of the authors, and do not reflect the official policy or position of IBM or the IBM Quantum team. We acknowledge AWS Cloud Credit for Research program that allowed us to perform experiments on Rigetti Aspen-M-3. We thank Ingo Roth and Martin Kliesch for discussions. We ackowlege use of Manim package \cite{manim} in preparation of Fig. \ref{fig:connectivity_graphs}.
\end{acknowledgments}

\section{Methods}

 \subsubsection{Distances between quantum measurements}\label{sec:background_distances}


To quantify quality of measurements' implementation, we use the worst-case (or operational) distance, which for POVMs $\M$ and $\P$ is defined as 
 \begin{align}
 \label{eq:operational_dist}
     \dist_{WC}\left(\M,\P\right) = \sup_{\rho} \tvd\left(\p(\M,\rho),\ \p(\P,\rho)\right) \ ,
 \end{align}
 where  $\tvd\left(\p(\M,\rho),\ \p(\P,\rho)\right)$ is Total Variation Distance between probability distributions induced by measurements $\M$ and $\P$ on $\rho$, which is defined as   
 \begin{align}\label{eq:tvd}
     \tvd\left(\p,\ \mathbf{q}\right) = 
     \frac{1}{2} \sum_{i} |p_i-q_i| \ .
 \end{align}
The operational distance quantifies the highest possible difference in statistics generated by two measurements over all possible states.

We also use a recently proposed alternative to the worst-case distance, the average-case distance \cite{MaciejewskiAverageCaseDistance,MaciejewskiAverageCaseDistanceProofs} defined as 
\begin{align} 
\label{eq:average_dist}
    \dist_{AV}\left(\M,\P\right) = \frac{1}{2d}\sum_{i}\sqrt{||M_i-P_i||_{\text{HS}}^2+\tr^2\left(M_i-P_i\right)}.
\end{align}
Due to the fact that $\dist_{AV}\left(\M,\P\right)\approx \expect{\psi}\  \tvd\left(\p(\M,\psi),\ \p(\P,\psi)\right)\ $ ( see \cite{MaciejewskiAverageCaseDistanceProofs} for details), the average-case distance quantifies average difference in statistics induced by two quantum measurements w.r.t. pure states.

\subsubsection{Lower bound on Coherence Strength}
\label{sec:app_cs_bound}
In this section Proposition \ref{prop:cs_lower_bound} on a lower bound on coherence strength of noisy POVM $\M$ is proved. 

\textbf{Proposition \ref{prop:cs_lower_bound}} \textit{For a POVM $\M$ and quantum states $|p\rangle, |q\rangle$ constructed from tensor product of Pauli $X,Y$ eigenstates the following bound holds} 
\begin{eqnarray}
    \mathrm{CS}(\M)^{AC} \geq \frac{\tvd \left( \Pr(p), \Pr(q)   \right)}{d \| \delta_{p,q} \|_{HS}},
\end{eqnarray}
\textit{where $\Pr(p)$ ($\Pr(q)$) is probability distribution obtained by measuring $|p\rangle$ ($|q\rangle$)  with $\M$, and $\delta_{p,q} =  \left( | p \rangle \langle p | - | q \rangle \langle q | \right ) $, and $d$ is  dimension of the corresponding Hilbert space.} 


\begin{proof}
Let us decompose effects of $\M$ into diagonal and off-diagonal part $M_i =M_i^{diag} + \Delta_i$. For a particular outcome $i$ it holds

\begin{eqnarray}
    &&\Pr(i|p) - \Pr(i|q) = \tr \left( M_i \left( | p \rangle \langle p | - | q \rangle \langle q | \right) \right) = \\&& \tr \left( \Delta_i \left( | p \rangle \langle p | - | q \rangle \langle q | \right ) \right) =\tr \left( \Delta_i \delta_{p,q} \right) \leq \| \Delta_i \|_{HS} \| \delta_{p,q} \|_{HS}, \nonumber
\end{eqnarray}
where the second equality follows from the fact that the states $| p \rangle \langle p |, | q \rangle \langle q |$ differ only on off-diagonal elements. Using this fact we can write 

\begin{eqnarray}
     &&\tvd \left( \Pr(p), \Pr(q)   \right) \leq \frac{1}{2} \sum_i \sqrt{ \| \Delta_i \|_{HS}^2} \sqrt{\| \delta_{p,q} \|^2_{HS}} \nonumber \\ && \leq \frac{1}{2}  \| \delta_{p,q} |_{HS} \sum_i  \sqrt{  \| \Delta_{i} |^2_{HS} + \tr \left( \Delta_{i}  \right)^2} \nonumber \leq \\ 
    && d \times \| \delta_{p,q} |_{HS} \times  \dist_{AV} (\M, \Phi_{\text{dep}}(\M) ),
\end{eqnarray}
where the last inequality follows from the definition of average case distance. By using definition of coherence strength $\mathrm{CS}(\M)^{AC}$  the bound is obtained.

\end{proof}

\begin{prop}
Consider $\hat{\mathrm{CS}}(\M)^{AC}$ that is an empirical estimation of $\mathrm{CS}(\M)^{AC}$. It holds that
\begin{eqnarray}
\label{eq:cs_experimental}
  \left|   \mathrm{CS}(\M)^{AC} -  \hat{\mathrm{CS}}(\M)^{AC}  \right| \leq  \frac{\epsilon_{p}+ \epsilon_{q}}{d}.
\end{eqnarray}
where  $\epsilon_{p}$, $\epsilon_{q}$ are upper bounds on distances between true distributions $\Pr(p), \Pr(q)$, and their empirical values $\hat{\Pr}(p), \hat{\Pr}(q)$.        
\end{prop}
\begin{proof}
From the fact that $\tvd \left( \Pr(p), \hat{\Pr}(p)   \right) \leq \epsilon_{p}$, and  $\tvd \left( \Pr(q),  \hat{\Pr}(q) \right) \leq \epsilon_{q}$ it follows that
\begin{eqnarray}
    && \left| \tvd \left( \Pr(p), \Pr(q)   \right) - \tvd \left( \hat{\Pr}(p), \hat{\Pr}(q)   \right) \right|    \leq \epsilon_{p} + \epsilon_{q}. 
\end{eqnarray}
From the above inequality one obtains the bound (\ref{eq:cs_experimental}).
\end{proof}

The probabilistic upper bounds on $\epsilon_{p}, \epsilon_{q}$ are obtained using the results of \cite{Weissman2003}:

\begin{eqnarray}
       \tvd\left( \Pr(p), \hat{\Pr}(p)   \right) \leq \epsilon^*_{p}= \sqrt{\frac{\log\rbracket{2^d-2}-\log{\rbracket{P_{ \text{err}}}}}{2 N_S}} ,
\end{eqnarray}
where $N_S$ is the number of shots used in the estimation, and  $P_{ \text{err}}$ denotes the error probability:  $\text{Pr}\rbracket{\tvd \left( \Pr(p), \hat{\Pr}(p)   \right) \geq \epsilon_p}$.

\onecolumngrid

\subsubsection{Details of experiments}
\label{app:experimental_details}
\begin{table}[!h]
\begin{tabular}{|c|c|c|c|c|c|c|}
\hline
Device      & \begin{tabular}[c]{@{}c@{}}number of \\  qubits\end{tabular} & \begin{tabular}[c]{@{}c@{}} experiment \\ date \end{tabular} & \begin{tabular}[c]{@{}c@{}} number of \\ DDOT circuits \end{tabular} & \begin{tabular}[c]{@{}c@{}}  number of \\ benchmark circuits \end{tabular} &  \begin{tabular}[c]{@{}c@{}}  number of \\ cohernce witness circuits \end{tabular} & \begin{tabular}[c]{@{}c@{}} number of \\ shots \end{tabular} \\ \hline
IBM Cusco & 127                                                                      & 29-09-2023   & 1200  & 300 & 88                & 10 000              \\ \hline
ASPEN--M--3     & 79                                                                         & 25-09-2023   & 300  & 50 & 16                 & 10 000              \\ \hline
\end{tabular}
\caption{Details of experiments on IBM's and Rigetti's devices.}
\label{tab:data_details}
\end{table}

\clearpage

\bibliography{ref}

\appendix
\section{Details of overlapping detector tomography}
In this section we provide detailed analysis of the Detector Overlapping Tomography protocol. In Sec. \ref{sub:dotmatrixel} we derive sampling complexity of the matrix entries of the reduced POVM effects. In Sec. \ref{sub:dotchoi} we use those results to derive sample complexity of a tomographic reconstruction of the Choi-Jamiolkowski state corresponding to the reduced measurement channel. Finally in Sec. \ref{sub:dotdistances} we present relations leading to Theorem 2 of the main text, which involve sample complexity of the Choi-Jamiolkowski state reconstruction and average, as well as worst-case distances.
\label{app:DOT}
\subsection{Sample complexity of reduced noise channel matrix elements estimation}
\label{sub:dotmatrixel}
Here derivation of the sampling complexity of estimating matrix elements of reduced POVM $\M_{\x_A,\y_A}^{A}$ is presented. The first step consists of preparation of input states to overlapping detector tomography procedure (for detailed description see Algorithm \ref{alg:DOT_rand}). To this end a random set $\Y^{N_{Mes}}$ of $N_{Mes}$ strings of length $N$ is generated. Each string from this set encodes one input state. If a six letter alphabet was used to generate the set then the input states consist of eigenstates of Pauli operators(this is full DOT protocol), whereas a two letter alphabet corresponds to computational basis states for DDOT protocol. Results of all measurements are stored in the set $\X^{N_{Mes}}$. For a chosen $k$-qubit subset $A$, $\hat{M}^{A}_{\mathbf{x}_A,\mathbf{y}_A}$  corresponding to a measurement process, in which a state $\y_A$ on $A$ was prepared and an outcome $\x_A$ observed, is estimated as:  
\begin{eqnarray}
    \label{eq:Lambda_matrix_el}
    \hat{M}_{\mathbf{x}_A,\mathbf{y}_A}^{A} = \frac{1 }{h_{\y_A}(\Y^{N_{Mes}})} \sum_{\y^i} q_{\x^i_{A},\y^i_A},
\end{eqnarray}
where 
\begin{equation}
q_{\x^i_A,\y^i_A} = 6^k \delta_{\x_A,\x^i_A} \delta_{\y_A,\y^i_A}.    
\end{equation}
is a random variable computed from experimental data, and
$h_{\y_A}(\Y^{N_{Mes}}) = \sum_{i=1}^{N_{Mes}} 1 \{ \y_A = \y^i_A \} = |\Y^{N_{Mes}}_{\y_A}|$, where $\Y^{N_{Mes}}_{\y_A} \subset \Y^{N_{Mes}}$ is a subset whose elements realize the chosen bit-string $\y_A$ on $A$.

The empirical average $f_{\x_A,\y_A}^{A}$ can be used to approximate $\M_{\x_A,\y_A}^{A}$ since the random variable $q_{\x^i_A,\y^i_A}$ is an unbiased estimator of $\M_{\x_A,\y_A}^{A}$: 
  \begin{eqnarray}
      &&\mathbb{E}_{\x,\y} q_{\x^i_A,\y^i_A}= \sum_{\x,\y}  \text{Pr}(\x,\y) q_{\x^i_A,\y^i_A} = \sum_{\x,\y} \text{Pr}(\y)   \text{Pr}(\x|\y) q_{\x^i_A,\y^i_A} =  
      \\&&  \sum_{\y_{\bar{A}}}\frac{1}{6^{N-k}} \sum_{\x_{\bar{A}}}\text{tr} \left( |\x_A \rangle \langle \x_A| \otimes |\x_{\bar{A}} \rangle \langle \x_{\bar{A}}|   \M \left( |\y_A \rangle \langle \y_A| \otimes |\y_{\bar{A}} \rangle \langle \y_{\bar{A}}|   \right) \right)  \nonumber = \\&&\text{tr}_A \left( |\x_A \rangle \langle \x_A|     \sum_{\x_{\bar{A}}} \langle \x_{\bar{A}}| \M \left( |\y_A \rangle \langle \y_A| \otimes \sum_{\y_{\bar{A}}}\frac{1}{6^{N-k}}|\y_{\bar{A}} \rangle \langle \y_{\bar{A}}|   \right) |\x_{\bar{A}} \rangle \right) \nonumber = \\&&
      \text{tr} \left( |\x_A \rangle \langle \x_A| \M^{A} \left( |\y_A \rangle \langle \y_A| \right) \right) = \M^{A}_{\x_A,\y_A},  \nonumber
  \end{eqnarray}
where we used the Born rule, the fact that in a QDOT experiment probability of randomly choosing a string $\y$ of $N$ length from a six  symbol alphabet is $\text{Pr}(\y_{\bar{A}})=\frac{1}{6^{N}}$ (in the case of DDOT protocol the alphabet has two symbols so $\frac{1}{2^{N}}$), and  $\sum_{\y_{\bar{A}}} |\y_{\bar{A}} \rangle \langle \y_{\bar{A}}| = 3^{N-k} I_{\bar{A}} $ ($\sum_{\y_{\bar{A}}} |\y_{\bar{A}} \rangle \langle \y_{\bar{A}}| =  I_{\bar{A}}$ for DDOT). \\
We also compute variance of the random variable $q_{\x^i_A,\y^i_A}$ 
\begin{eqnarray}
      &&\mathbb{V}_{\x,\y} q_{\x^i_A,\y^i_A} = 6^{k} \M^{A}_{\x_A,\y_A} - \left(\M^{A}_{\x_A,\y_A} \right)^2 
\end{eqnarray}

  We are now in position to prove Theorem 1 from the main text.

\textbf{Theorem 1 } \textit{Experimental DOT data allow to estimate $M_{\x_A,\y_A}^{A}$ (matrix elements of a POVM effect) for all possible choices of the $k$ qubit subset A and within $\epsilon$ error if the number $N_{Mes}$ of circuits used in the experiment scales as $O\left(\frac{k 6^k \log (N)}{\epsilon^2} \right)$. In the case of DDOT the scaling is $O\left(\frac{k 2^k \log (N)}{\epsilon^2} \right)$   }

\textit{Proof.} 
We start by bounding the probability of failure of joint estimation of matrix elements in a QDOT experiment:
\begin{eqnarray}
\label{eq:DOT_union_bound}
    \text{Pr} \left[ \exists_{\x_A,\y_A} \left| \hat{M}_{\x_A,\y_A}^{A} - M_{\x_A,\y_A}^{A} \right| \geq \epsilon \right].
\end{eqnarray}
We first use the fact that the total probability of an experiment can be decomposed into probability of choosing a particular circuit $\tilde{\y}$ and probability of obtaining a particular result from set $\X^{N_{Mes}}$ conditioned on the circuit choice. We firstly use Boole's inequality 
\begin{eqnarray}
   \text{Pr}_{\X^{N_{Mes}}} \left[ \exists_{\x_A} \left| \hat{M}_{\x_A,\y_A}^{A} - M_{\x_A,\y_A}^{A} \right| \geq \epsilon | \tilde{y} \right] \leq 2^k \binom{N}{k}    \text{Pr}_{\X^{N_{Mes}}} \left[ \left| \hat{M}_{\x_A,\y_A}^{A} - M_{\x_A,\y_A}^{A} \right| \geq \epsilon | \tilde{y} \right].
\end{eqnarray}
The right hand side can be bounded using Hoeffding's inequality
\begin{eqnarray}
    \text{Pr}_{\X^{N_{Mes}}} \left[ \left| \hat{M}_{\x_A,\y_A}^{A} - M_{\x_A,\y_A}^{A} \right| \geq \epsilon | \tilde{y} \right] \leq  2 \exp \left(-2 \epsilon^2 |\Y^{N_{Mes}}_{\y_A}| \right)
\end{eqnarray}

Combination of this two steps gives 
\begin{eqnarray}
   \text{Pr}_{\X^{N_{Mes}}} \left[ \exists_{\x_A} \left| \hat{M}_{\x_A,\y_A}^{A} - M_{\x_A,\y_A}^{A} \right| \geq  \epsilon | \tilde{y} \right] \leq 2 \times 2^k \binom{N}{k} \exp \left(-2 \epsilon^2 |\Y^{N_{Mes}}_{\y_A}| \right).
\end{eqnarray}
Our aim now is to estimate the cardinality of the set $\Y^{N_{Mes}}_{\y_A}$. We firstly use the fact that the experiments are independent and $|\Y^{N_{Mes}}_{\y_A}|$ factorizes into individual contributions. 
\begin{eqnarray}
\label{eq:DOT_complexity_derivation}
&&\text{Pr}_{\Y^{N_{Mes}}}\text{Pr}_{\X^{N_{Mes}}} \left[  \exists_{\x_A} \left| \hat{M}_{\x_A,\y_A}^{A} - M_{\x_A,\y_A}^{A}  \right| \geq \epsilon \right] \leq 2 \times  2^k \binom{N}{k}2  \text{Pr}_{\Y^{N_{Mes}}} \exp \left(-2 \epsilon^2 |\Y^{N_{Mes}}_{\y_A}| \right) = \\&& \nonumber 2 \times  2^k \binom{N}{k} \prod_{i=1}^{N_{{Mes}}}\text{Pr}_{\Y^{N_{Mes}}}  \exp \left(-2 \epsilon^2 1 \{ \y_A = \y^i_A \} \right)
\end{eqnarray}
Subsequently we can evaluate expectation value
\begin{eqnarray}    
&&\text{Pr}_{\Y^{N_{Mes}}}\text{Pr}_{\X^{N_{Mes}}} \left[  \exists_{\x_A} \left| \hat{M}_{\x_A,\y_A}^{A} - M_{\x_A,\y_A}^{A}  \right| \geq \epsilon \right] \leq  2 \times  2^k \binom{N}{k} \prod_{i=1}^{N_{{Mes}}}\text{Pr}_{\Y^{N_{Mes}}}  \exp \left(-2 \epsilon^2 1 \{ \y_A = \y^i_A \} \right) \leq \\ && \nonumber  2 \times 2^k \binom{N}{k}  \left(\frac{1}{6^k} \exp \left(-2 \epsilon^2 \right) + \left( 1 - \frac{1}{6^k} \right) \right)^{N_{Mes}}  = 2 \times 2^k \binom{N}{k} \left( 1 - \frac{1}{6^k}\left( 1 - \exp(-2\epsilon^2)\right) \right)^{N_{Mes}} \leq \\ && \nonumber 2 \times 2^k \binom{N}{k}  \exp\left(- \frac{(1-\exp(-2\epsilon^2))N_{Mes}}{6^k} \right),
\end{eqnarray}
and, by introducing an upper bound for confidence we find that the number of circuits  required to estimate matrix elements of POVM effects with $\epsilon$ accuracy and $1-\delta$ confidence scales as 
  \begin{eqnarray}
      N_{Mes} = \frac{6^k}{1-\exp(-2\epsilon^2)} \left[ (k+1) \log 2  + \log \binom{N}{k} +\log\frac{1}{\delta}  \right]. \nonumber \\
  \end{eqnarray}
 In the case of a DDOT  one needs to replace the factor $6^k$ with a factor $2^k$ in Eq. (\ref{eq:DOT_union_bound}), and accordingly in Eq. (\ref{eq:DOT_complexity_derivation}). The sample complexity is then

  \begin{eqnarray}
      N_{Mes}  = \frac{2^k}{1-\exp(-2\epsilon^2)} \left[ (k + 1)\log 2 +  \log \binom{N}{k} +\log\frac{1}{\delta}  \right], \nonumber \\
  \end{eqnarray}
which gives the scaling reported in the main text \qedsymbol

\begin{algorithm*}
\small
 \caption{Single-shot estimation POVM estimation algorithm \label{alg:DOT_rand}$\qquad \ $}
  \RaggedRight\textbf{Input}:\\
  \RaggedRight $N$: number of all qubits\\
  \RaggedRight $k$: size of the marginal\\
\begin{enumerate}
    \settowidth\tablen{$\quad$}
    \item Generate a $N_{Mes}$ strings $\y^i$ of length $N$, whose entries are randomly chosen from 6 elements (2 elements for DDOT experiment). The set of all those strings is denoted as $\Y^{N_{Mes}}$.  
	\item \textbf{For} each string $\y^i$, \textbf{do}:\\
    \tabbox{$\quad$} Construct a quantum circuit that prepares $| \y^i \rangle$ and measure it in the computational basis. \\    \tabbox{$\quad$} Store the result $\x^i$ (the set of all measurement outcomes is denoted as $\X^{N_{Mes}}$).
\item    Choose a $k$ qubit subset $A$, as well as fix the input and the output strings $\x_A, \y_A$. Calculate a random variable
    \begin{equation}
    \label{eq:random_q}
    q_{\x^i_A,\y^i_A} = \left\{ \begin{array}{ll}
1 & \textrm{when $\x_A = \x^i_A$ and $\y_A = \y^i_A$   }\\
0 & \textrm{otherwise}
\end{array} \right.
    \end{equation}
\item Output empirical average as an approximation of $M_{\x_A,\y_A}^{A}$  
\begin{eqnarray}
    \hat{M}_{\x_A,\y_A}^{A} = \frac{1 }{h_{\y_A}(\Y^{N_{Mes}})} \sum_{\y^i_{A}} q_{\x^i_{A},\y^i_A},
\end{eqnarray}
with
\begin{equation}
h_{\y_A}(\Y^{N_{Mes}}) = \sum_{i=1}^{N_{Mes}} 1 \{ \y_A = \y^i_A \} = |\Y^{N_{Mes}}_{Y_A}|,    
\end{equation}
where $\Y^{N_{Mes}}_{\y_A} \subset \Y^{N_{Mes}}$ is a subset whose elements realize the chosen bit-string $\y_A$ on $A$.
\end{enumerate}
\end{algorithm*}

Due to technological limitations of NISQ devices a modified version of the protocol, in which quantum circuits are reused many times, is more practical. In such a setting for each choice of $\y^i$ the circuit will be implemented and measured  $S$ times.
The estimator in a given shot $s$ is constructed as 
\begin{equation}
q_{\x^{i,s}_A,\y^i_A} = 6^k \frac{1}{S} \sum_{s} \delta_{\x_A,\x^{i,s}_A} \delta_{\y_A,\y^i_A},    
\end{equation}
and then averaged over all shots in a given round
\begin{equation}
    q^S_{\x^{i}_A,\y^i_A} = \frac{1}{S} \sum_s q_{\x^{i,s}_A,\y^i_A}.
\end{equation}
For $S=1$ one recovers the standard protocol discussed above. To analyze performance of the modified protocol we will utilize results obtained in the context of classical shadows \cite{helsen2022thrifty,Zhou2023multishot}, where analogous multi-shot setting has been considered. Analogously to \cite{helsen2022thrifty,Zhou2023multishot} we use the Law of Toatal Variance to obtain the following relation  

\begin{eqnarray}
 \mathbb{V}_{\x,\y} \left(   q^{S}_{\x^{i}_A,\y^i_A} \right) = \frac{1}{S} \mathbb{V} \left(q^{S=1}_{\x^{i}_A,\y^i_A} \right) + \frac{S-1}{S} \mathbb{V}_{\y} \left( \mathbb{E}_{\x} \left( q^{S=1}_{\x^{i}_A,\y^i_A} | \y \right) \right),
\end{eqnarray}
where the first term on the r.h.s. is the variance of the estimator for the protocol without multiple shots, and the second is caused by the fact thata a single measurement setting is used multiple times. Let us firstly compute the conditional expectation value
\begin{eqnarray}
     \mathbb{E}_{\x} \left( q^{S=1}_{\x^{i}_A,\y^i_A} | \y \right) = 6^{k}   \sum_{\x_{\bar{A}}}\text{tr} \left( |\x_A \rangle \langle \x_A| \otimes |\x_{\bar{A}} \rangle \langle \x_{\bar{A}}|   \M \left( |\y_A \rangle \langle \y_A| \otimes |\y_{\bar{A}} \rangle \langle \y_{\bar{A}}|   \right) \right) \delta_{\y_A,\y_A^i}.
\end{eqnarray}
Variance over $\y$ can be then bounded as

\begin{eqnarray}
    \mathbb{V}_{\y} \left(  \mathbb{E}_{\x} \left( q^{S=1}_{\x^{i}_A,\y^i_A} | \y \right)  \right) =  \frac{6^{k}}{6^{N-k}} \sum_{\y_{\bar{A}}} \text{tr} \left( |\x_A \rangle \langle \x_A| \otimes  \sum_{\x_{\bar{A}}} |\x_{\bar{A}} \rangle \langle \x_{\bar{A}}|   \M \left( |\y_A \rangle \langle \y_A| \otimes |\y_{\bar{A}} \rangle \langle \y_{\bar{A}}|   \right) \right)^2 \leq 6^k.
\end{eqnarray}
This rough estimation already shows that the multi-shot setting does not lead to increased variance of the estimator, and more detailed calculations show that there exit circuits for which this bound is attained. As a result, it is not possible to utilize the multi shot scenario to decrese the estimator's variance.     

\subsection{Sample complexity of reduced noise channel tomographic reconstruction}
Here we derive sample complexity of reconstruction of reduced measurement Choi-Jamiolkowski states for all possible choices of $k$ qubits subsets.
\label{sub:dotchoi}
\subsubsection{Preliminaries: Choi-Jamiołkowski isomorphism}
In order to describe noisy quantum measurements we use the formalism of quantum channels. In general, a quantum channel $\Lambda(\cdot)$ is a completely positive trace preserving (CPTP) map $\mathcal{L}: \mathbb{C}^{d \times d} \rightarrow \mathbb{C}^{d \times d}$ mapping operators to operators. A quantum channel $\Lambda(\cdot)$  can be conveniently represented as quantum state, the so-called Choi state, via Choi-Jamiołkowski isomorphism  
\begin{eqnarray}
\label{eq:choi_definition}
    \mathcal{J}_{\Lambda} = (\Lambda(\cdot) \otimes I_d ) | \Omega \>\< \Omega|_{XY},
\end{eqnarray}
where $| \Omega \>_{XY} = \frac{1}{\sqrt{d}} \sum_{i}^d |ii\>_{XY}$  is the maximally entangled state. The action of a quantum channel on a quantum state $\rho$ is then expressed as
\begin{eqnarray}
    \Lambda(\rho) = \text{tr}_{Y}\left( \mathcal{J}_{\Lambda} \left( I \otimes \rho^T \right)  \right),
\end{eqnarray}
where the partial trace is performed over the second subsystem. 

The measurement process is a special instance of a quantum-classical quantum channel of a form
\begin{eqnarray}
\label{eq:measurement_channel}
    \mathcal{M}(\rho) = \sum_{i} \text{tr}\left( M_{i} \rho \right) |i \rangle \langle i| ,
\end{eqnarray}
where the post-measurement state is diagonal and according to Eq. (\ref{eq:choi_definition}) its Choi state is 
\begin{eqnarray}
\label{eq:ChoiMeasurement}
   \mathcal{J}_{\mathcal{M}} = \frac{1}{d} \sum_{i} |i \rangle \langle i| \otimes M_i^T
\end{eqnarray}
\subsubsection{Projected least-squares estimation of the reduced measurement channel}
 
Here we prove a Lemma regarding sample complexity of Choi-Jamiołkowski state reconstruction.  Our starting point is the realization that, using relations between quantum channels and their Choi states, $\M_{\x_A,\y_A}^{A}$ can be rewritten as 

\begin{eqnarray}
   \M^{A}_{\x_A,\y_A} = \text{tr} \left( |\x_A \rangle \langle \x_A| \M^{A} \left( |\y_A \rangle \langle \y_A| \right) \right) = \text{tr} \left( |\x_A \rangle \langle \x_A|  \mathcal{J}_{\M^A}  \left( I_A \otimes |\y_A \rangle \langle \y_A|^T_{A'}   \right) \right)
\end{eqnarray}

Using the results of \cite{SurawyStepneyQPT} we find that an unbiased estimator to the channels Choi state is
\begin{eqnarray}
\label{eq:choi_estimator}
     \mathcal{\hat{J}}_{\M^A}^{LS}  = \frac{1}{3^k\times 2^k} \sum_{\x_A,\y_A} \hat{\M^A}_{\x_A,\y_A} |\x_A \> \<\x_A| \otimes \left( 3 |\y_A \> \<\y_A| -I_{2^k} \right), 
\end{eqnarray}
where $|\y_A \> \<\y_A| $ is the state prepared as an input to the DOT protocol, which consists of a tensor product of single qubit Pauli eigenstates.

The requirements of a proper quantum map to be completely positive and trace preserving are equivalent to  positive semidefiniteness of the map's Choi state and the partial trace condition $\text{tr}_1 \left( \mathcal{J}_{\Lambda} \right) = I_d/d$. There is no guarantee that the estimator $\mathcal{\hat{J}}_{\M^A}^{LS}$ fulfills those conditions, and sometimes a projection on the set of $CPTP$ is required, for details see \cite{SurawyStepneyQPT}. Here we denote the projected estimator as $\mathcal{\hat{J}}_{\Lambda^A}^{PLS}$.

\textbf{Lemma 1}  \textit{Experimental DOT data allow to estimate Choi-Jamiolkowski state of the reduced noise matrix $\rho_{\mathcal{M}^A}$ for all possible choices of the $k$ qubit subset A and the  within $\epsilon$ error if the number $N_{Mes}$ of circuits used in the experiment as   $O\left(\frac{k 6^k \log (N) }{\epsilon^2} \right)$   } \\

\textit{Proof.} 
We start by bounding probability of a joint estimation of Choi-Jamiolkowski states of reduced measurement channels  
\begin{eqnarray}
\label{eq:DOT_union_bound}
    \text{Pr} \left[ \exists_{A} \left\| \mathcal{\hat{J}}_{\M^A}^{PLS} - \mathcal{J}_{\M^A} \right\|_{HS} \geq \epsilon \right].
\end{eqnarray}
From the union bound it follows that 
\begin{eqnarray}
\label{eq:ChoiNormBound}
        \text{Pr} \left[ \exists_{A} \left\| \mathcal{\hat{J}}_{\M^A}^{PLS} - \mathcal{J}_{\M^A} \right\|_{HS} \geq \epsilon \right] \leq  \binom{N}{k}\text{Pr} \left[ \left\| \mathcal{\hat{J}}_{\M^A}^{PLS} - \mathcal{J}_{\M^A}^{PLS} \right\|_{HS} \geq \epsilon \right]. 
\end{eqnarray}
Subsequently one adapts the results of \cite{SurawyStepneyQPT} - Eq. (16) in the main text. Taking into account that the measurement channels are of the quantum-classical form (the output of the channel is a mixture of orthogonal projectors with probabilities corresponding to probabilities of observing a corresponding outcome, so in Eq. (16) one can substitute $3^{2k}$ with $3^{k}$ as the full tomography of the output state of the channel is not need, one finds    that 
\begin{eqnarray}
\label{eq:ChoiComplexity}
    \text{Pr} \left[  \left\| \mathcal{\hat{J}}_{\M^A}^{PLS} - \mathcal{J}_{\M^A}  \right\|_{HS} \right] \leq 2^{2k} \exp \left[-\frac{3 N_{Mes} \epsilon^2}{64 \times 6^{k}} \right] ,
\end{eqnarray}
and, by introducing an upper bound on failure $\delta$ finally  we find that
\begin{eqnarray}
    N_{Mes}= \frac{6^k}{\epsilon^2} \frac{64}{3} \left[ \log \binom{N}{k} +2k \log 2 + \log \left( \frac{1}{\delta} \right) \right].
\end{eqnarray}
\qedsymbol
 
 \subsection{Relations between Hilbert-Schmidt norm and average-case, and worst case distances}
 \label{sub:dotdistances} 
In this section we bound average-case and worst case distance between two POVMs by the Hilbert-Schmidt norm of their respective Choi-Jamiolkowski states. This allows us to prove Theorem 2 of the main text.

We start by proving the following lemma regarding the averadge case distance: 
\begin{lem}
 \textit{For $d$ outcome POVMs $\M, \N$  the following relation between their average case and Hilbert-Schmidt distances holds} \\ 
\begin{eqnarray}
     \dist_{AC} \left( \M, \N \right) \leq \frac{\sqrt{d+1}}{2}  || \hat{J}_{\M}^{PLS} - \mathcal{J}_{\N} ||_{\text{HS}}.
\end{eqnarray}
\end{lem} 

\textit{Proof.-} \\

Average case distance is defined as 
 \begin{eqnarray}
 \label{ap:avdefproof}
 \dist_{AC} \left( \M, \N \right) = \frac{1}{2d} \sum_{k=1}^{d}  \sqrt{|| \M_k - N_k||_{\text{HS}}^2+\tr^2\left(M_{k}-N_k \right)}
 \end{eqnarray}

We proceed by bounding the second term on the right hand side. Let us considers decomposition of the POVM elements into a sum
 \begin{equation}
     M_k = \text{tr}\left( M_k  \right) \frac{I}{d} +  M_{k,0}, 
 \end{equation}
where $M_{i,0}$ is a traceless part of $M_{i}$. This allows to show that
\begin{equation}
    \tr^2\left(M_{k}-N_k\right) \leq d ||M_{k}-N_k||^2_{\text{HS}}.
\end{equation}
The above inequality is now inserted into the definition of the average case distance (Eq. (\ref{ap:avdefproof})) to get 
\begin{eqnarray}
     \dist_{AC} \left( M^{A}_{PLS}, M^{A} \right) \leq \frac{\sqrt{d+1}}{2d} \sum_{k=1}^d  ||M^{A}_{PLS,k}-M^{A} _k||_{\text{HS}}.
\end{eqnarray}
Subsequently we use  Cauchy–Schwarz inequality to obtain the bound
\begin{eqnarray}
     \dist_{AC} \left( \M, \N \right) \leq \frac{\sqrt{d(d+1)}}{2d} \sqrt{  \sum_{k=1}^d  ||\M_k-\N_k||^2_{\text{HS}} }.
\end{eqnarray}
Now use the explixit form of a measurement channel Choi-Jamiolkowski state (Eq. (\ref{eq:ChoiMeasurement})) to obtain  
\begin{eqnarray}
     \dist_{AC} \left( \M, \N \right) \leq \frac{\sqrt{d+1}}{2}  || \hat{J}_{\M} - \mathcal{J}_{\N} ||_{\text{HS}},
\end{eqnarray}
what is the statement of Lemma
\qedsymbol.

A similar reasoning is used in the case of worst-case distance. In  particular, the key relation here is 
\begin{eqnarray}
    \dist_{WC} \left( \M, \N \right) \leq \frac{1}{2} \sum_{k=1}^{d} ||M_{k}-N_k||_{\text{HS}} \leq \frac{d}{2} || \hat{J}_{\M} - \mathcal{J}_{\N} ||_{\text{HS}}.  
\end{eqnarray}

We can now apply Lemma 1 to the case of reconstruction of Choi-Jamiolkowski states of reduced  measurement channels $\M^A_{PLS}$ discussed in the previous subsection. Using Eq. (\ref{eq:ChoiComplexity}) one obtains Theorem 2 of the main text.

\subsection{Noise matrices estimation - Finite size effects }
Here we analyze finite size effects in estimation of noise matrix $\Lambda$ acting on marginal probability distributions. To this end we firstly define worst-case distance between stochastic maps

\textbf{Definition} \textit{Let $\Lambda$, $\Gamma$ be two stochastic maps acting on $d$ dimensional probability distributions. The worst-case distance between $\Lambda$, $\Gamma$ is defined as}
\begin{equation}
\| \Lambda - \Gamma  \|_{WC} \equiv \max_{p} \| \left( \Lambda- \Gamma \right) p  \|_{1} = \max_{l=1,\ldots, d}    \|  \left( \Lambda- \Gamma \right) e_l  \|_{1},
\end{equation}
where the maximum is taken with respect to all possible $d$ dimensional probability distributions $p$, and from convexity it follows that it is attained on one of the extremal probability distributions $e_l$, which are of a form $e_{l,i}=\delta_{li}$.

We now provide a bound on estimation of a noise matrix $\Lambda$ from the measurement statistics.   

\textbf{Lemma 2} \textit{The worst case distance between a noise matrix $\Lambda$ and its estimate $\hat{\Lambda}$ is bounded by}

\begin{eqnarray}
   \| \hat{\Lambda} - \Lambda  \|_{WC} \leq \epsilon^* =  \sqrt{\frac{\log \left(2^k(2^{2^k}-2) \right) - \log P_{err} }{2N_s}},
\end{eqnarray}
\textit{where $k$ is the number of qubits, on which $\Lambda$ acts, $N_{S}$ is the number of samples used to estimate $\hat{\Lambda}$, and $P_{err}$ is the probability that the bound does not hold. }

\textit{Proof.-}

Using definition of the worst case distance for stochastic matrices we have 
\begin{eqnarray}
    \| \hat{\Lambda} - \Lambda  \|_{WC} =  \max_{l=1,\ldots, d} \text{TVD} \left( \hat{\Lambda} e_l ,\Lambda e_l\right) =  \max_{l=1,\ldots, d} \text{TVD} \left( \hat{p}^{\Lambda}_l ,p^{\Lambda}_l\right) ,
\end{eqnarray}
where $\hat{p}^{\Lambda}_l \equiv \hat{\Lambda} e_l$ is estimate of probability distribution corresponding to the action of the noise matrix on $l$-th extremal probability distribution $e_l$. As a result, computation of the worst case distance between $\hat{\Lambda}$ and $\Lambda$ has been translated into maximization of Total Variational Distance of corresponding probability distributions. Therefore, we can use results of \cite{Weissman2003}, where the probability of estimated distribution  being $\epsilon$ close to the true distribution was bounded as  
\begin{eqnarray}
    \text{Pr} \left[ \text{TVD}\left( \hat{p}^{\Lambda}_l, p^{\Lambda}_l \right) \geq \epsilon \right] \leq (2^{2^k}-2) )\exp(-2 N_S \epsilon^2),
\end{eqnarray}
where $N_S$ denotes the number of samples used to estimate $\hat{p}^{\Lambda}_l$, $k$ is the number of qubits and hence $2^k$ is dimension of space of $p^{\Lambda}_l$. To obtain the bond on the worst-case distance between $\hat{\Lambda}$ and $\Lambda$, one needs to perform optimization over $l$, i.e. all pairs of probability distributions and their estimates. To this end we use union bound and write
\begin{eqnarray}
\label{eq:HPbound}
   \text{Pr} \left[\| \hat{\Lambda} - \Lambda  \|_{WC}  \geq \epsilon \right] =   \text{Pr} \left[ \max_{k=1,\ldots, d} \text{TVD}\left( \hat{p}^{\Lambda}_l, p^{\Lambda}_l \right) \geq \epsilon \right] \leq 2^k (2^{2^k}-2)\exp(-2 N_S \epsilon^2),
\end{eqnarray}
where the last inequality was obtained using the fact that the $d=2^k$.

We will now consider the above probability to be fixed and equal $P_{err}$. Appropriate rewriting of Eq. (\ref{eq:HPbound}) allows to arrive at the statement of Lemma 2  
\begin{eqnarray}
   \| \hat{\Lambda} - \Lambda  \|_{WC} \leq \epsilon^* =  \sqrt{\frac{\log \left(2^k(2^{2^k}-2) \right) - \log P_{err} }{2 N_S}}.
\end{eqnarray} \qedsymbol

\section{Correlations coefficients}\label{sec:app:correlation_coefficients}
The aim of this Section is to investigate some properties of quantum correlation coefficients (Sec. ~\ref{sec:app:correlation_coefficients_properties} ), and present a method to compute them (Sec. ~\ref{sec:app:correlation_coefficients_calculation}.)   

In the main text, we defined quantum correlations coefficient for qubit pairs. More generally quantum correlation coefficients for two distinguished subsystems $A$, $B$ can be defined as

\begin{eqnarray}
\label{eq:app:correlation_coefficient_general}
    c^{WC/AC}_{B \rightarrow A} = \sup_{\rho_{B}, \sigma_{B}} \dist_{WC/AC} \left( \M^{A,\rho_{B}}, \M^{A,\sigma_{B}} \right)\     ,
\end{eqnarray}
where $\dist_{WC}$ indicates worst-case (Eq.~\eqref{eq:operational_dist}), and $\dist_{ac}$  average-case distance (Eq.~\eqref{eq:average_dist}) distance.

\subsection{Properties}\label{sec:app:correlation_coefficients_properties}

We begin by showing that for diagonal POVMs the supremum involved in the definition of the correlation coefficients is attained by computational basis states.

\begin{prop}
\label{prop:app:classcial_noise_correlations_general}
Consider a POVM $\M$ with \emph{diagonal} effects, and its reduced versions $\M^{A, \rho_{B}}$, $\M^{A,\sigma_{B}}$  acting on subsystem $A$ (for definition see  Eq.~\eqref{eq:reduced_povm_definition}). 
Then we have
\begin{align}\label{eq:app:wc_coefficient_classical}
    \sup_{\rho_{B},\sigma_{B}} \dist_{WC/AC}(\M^{A,\rho_{B}},\  \M^{A,\sigma_{B}}) =
    \sup_{\x_B,\ \y_B} \dist_{WC/AC}(\M^{A,\ketbra{\x}{\x}_{B}},\  \M^{A,\ketbra{\y}{\y}_{B}}),
\end{align}
where $\ketbra{\x}{\x}_{B}$ and $\ketbra{\y}{\y}_{B}$ denote classical, computational-basis states labeled by bitstrings $\x_{B}$ and $\y_{B}$.
\begin{proof}
The proof consists of two steps. 
Firstly, from linearity of partial trace, and from convexity of both distances \cite{MaciejewskiAverageCaseDistanceProofs}, it follows that the supremum is attained for pure states. 
Thus the maximization can be constrained to pure states on $B$. Secondly, the fact that POVM effects are diagonal implies that its reduced effects are also diagonal. Therefore in maximization over pure states on $B$ it suffices to consider only diagonal states (this follows directly from Eq.~\eqref{eq:reduced_povm_definition}).
Thus we restricted maximization to pure, diagonal states, hence states from computational basis.
\end{proof}
\end{prop}
In the case of two-qubit diagonal POVMs Proposition \ref{prop:app:classcial_noise_correlations_general} implies that supremum is attained for $\ketbra{0}{0}$, $\ketbra{1}{1}$ states on the second qubit.

Let us now consider that a POVM $\M$ is related to a projective computational-basis measurement $\P$ via a stochastic map $\Lambda$.

We now provide definition of correlation coefficients in terms of a stochastic map $\Lambda$, which transforms a projective computational-basis measurement $\P$ into .  
In a previous work \cite{Maciejewski2021modeling} correlation coefficients were defined in terms of stochastic maps $\Lambda$ acting on a computational-basis projective measurement $\P$.


\begin{property}\label{prop:app:stochastic_maps_correlations}
Consider arbitrary two-qubit POVM $\M$ related to computational-basis measurement $\P$ via a stochastic transformation $\Lambda$: $M_{\x_{AB}} = \sum_{\x_{A}' \x_{B}'} \Lambda_{\x_{A} \x_{B} \x_{A}'\x_{B}'}P_{\x_{A}' \x_{B}'}$.
Then we have 
\begin{align}\label{eq:app:wc_stochastic_maps}
    c_{B\rightarrow A}^{WC} = \frac{1}{2}|| \Lambda^{A,\ketbra{0}{0}_B}-\Lambda^{A,\ketbra{1}{1}_B}||_{1\rightarrow 1}
\end{align}
\begin{align}\label{eq:app:ac_stochastic_maps}
    c_{k\rightarrow l}^{AC} = \frac{1}{2}\sqrt{\frac{1}{2}|| \Lambda^{A,\ketbra{0}{0}_B}-\Lambda^{A,\ketbra{1}{1}_B}||_{HS}^{2}+\left(\tr{(\Lambda^{A,\ketbra{0}{0}_B}-\Lambda^{A,\ketbra{1}{1}_B})}\right)^2}.
\end{align}
The norm $||L||_{1\rightarrow1}$ is a norm induced by $l1$ distance and it can be written as $||L||_{1\rightarrow 1}$ = $\max_{l} \sum_{k}|L_{kl}|$ (maximal $l1$ norm over columns of $A$). $\Lambda^{A,\ketbra{\x}{\x}_B}  $ is a stochastic map relating reduced POVM on qubit $A$ and a single-qubit computational basis measurement, provided that qubit $B$ was in state $\ketbra{x}{x}$, with entries $\Lambda^{A,\ketbra{\tilde{\x}}{\tilde{\x}}_B}_{\x_A, \x_A'} \equiv \sum_{\x_B} \Lambda_{\x_{A} \x_{B} \x_{A}'\tilde{\x}}$. 
We note that Eq.~\eqref{eq:app:wc_stochastic_maps} recovers definition of classical correlations coefficients introduced in Ref.~\cite{Maciejewski2021modeling}.
\begin{proof}
First step of the proof is to use Proposition~\ref{prop:app:classcial_noise_correlations_general} that allows us to consider only classical states on qubit $B$. This implies that the solution of maximization is given by marginal POVMs on qubit $A$ conditioned on qubit $B$ being in either state $\ketbra{0}{0}$ or $\ketbra{1}{1}$, so only stochastic maps corresponding to those input states need to be considered. Taking this into account explicit calculations of  both distances on matrix elements of reduced POVMs using Eq.~\eqref{eq:app:wc_combinatorial} yield desired relations.
\end{proof}
\end{property}

\begin{property}\label{prop:app:lower_bounds_correlations}
Consider two-qubit POVM $\M$ (acting on qubit $A$ and $B$).
Consider its reduced POVMs $\M^{A,\rho_B}$ and $\M^{A,\sigma_B}$ on qubit $A$ conditional on qubit $B$ being in state $\rho_B$ and $\sigma_B$, respectively.
Define maximally depolarizing channel $\Phi_{\mathrm{dep}}$ that acts on POVM $\M\rightarrow \Phi_{\mathrm{dep}}(\M)$ by changing its effects to $\M_i\rightarrow\Phi_{\mathrm{dep}}(\M_i)=\mathrm{diag}\left(\M_i\right)$, where $\mathrm{diag}\left(L\right)$ denotes diagonal part of operator $L$.

Then we have
\begin{eqnarray}\label{eq:app:wc_coefficient_classical}
    c_{B\rightarrow A}^{WC/AC} = &&\sup_{\rho_B, \sigma_B} \dist_{WC/AC}(\M^{A,\rho_B},\  \M^{A,\sigma_B}) \geq \sup_{\rho_B, \sigma_B} \dist_{WC/AC}\left(\Phi_{\mathrm{dep}}\left(\M^{A,\rho_B}\right),\  \Phi_{\mathrm{dep}}\left(\M^{A,\sigma_B}\right)\right) = \\ &&\dist_{WC/AC}(\M^{A,\ketbra{0}{0}_B},\  \M^{A,\ketbra{1}{1}_B})  ,
\end{eqnarray}
\begin{proof}
The proof has two steps.
The first inequality follows from data processing inequality. Then the equality on RHS in follows from Property~\ref{prop:app:classcial_noise_correlations_general} proved above.
\end{proof}
\end{property}
The depolarized version of POVM can be considered as a classical part of the noise -- this is the part of POVM that can be related to computational-basis projective measurement $\P$ via some stochastic map $\Lambda$ \cite{Maciejewski2020mitigation}.
The above thus means that total, quantum correlations are lower bounded by classical correlations in readout noise -- as one would expect.

\subsection{Calculating correlations coefficients}\label{sec:app:correlation_coefficients_calculation}

In previous Section~\ref{sec:app:correlation_coefficients_properties} we have shown that if the noise is \emph{classical} (i.e., reduced POVMs are diagonal), then both worst-case and average-case correlations coefficients can be computed directly from stochastic maps describing reduced noise on qubit $i$, conditioned on qubit $j$ being in classical state (Eqs.~\ref{eq:app:wc_stochastic_maps} and \ref{eq:app:ac_stochastic_maps}).
When noise is non-classical, then one needs to perform optimization over all pure (see proof of Property~\ref{prop:app:classcial_noise_correlations_general}) quantum states of the functions that cannot be optimized directly using exact methods. Note, however, that classical correlation coefficients (i.e., coefficients corresponding to diagonal part of the POVM) provide a natural lower bound for general, quantum correlation coefficients as follows from Property~\ref{prop:app:lower_bounds_correlations}.

We now consider problem of calculating those general coefficients. 
Let us start with considering optimization involved in Eq.~\eqref{eq:app:correlation_coefficient_general} when distance is worst-case.
Here we have (cf. Eq~\eqref{eq:correlation_coefficient_general})
\begin{align}
  c_{B\rightarrow A}^{WC} =   \sup_{\rho_{B}, \sigma_{B}}\ \dist_{WC}(\M_{A}^{\rho_{B}}, \M_{A}^{\sigma_{B}}) = \sup_{\rho_B, \sigma_B}\ ||(\M_{A}^{\rho_{B}})_0-(\M_{A}^{\sigma_{B}})_0||_{op} \ .
\end{align}

We start by rewriting optimization problem as follows
\begin{align}
    c_{B\rightarrow A}^{WC} &= \sup_{\rho_B,\sigma_B} ||(\M_{A}^{\rho_B})_0-(\M_{A}^{\sigma_B})_0||_{op} =  \sup_{\rho_B,\sigma_B} ||\tr_{B}{\left(\M_0\ \iden\otimes (\rho_B-\sigma_B)\right)}||_{op} = \sup_{\Delta} ||\tr_{B}{\left(\M_0\ \iden\otimes \Delta\right)}||_{op} \ .
\end{align}
In the above we defined we defined coarse-grained (over qubit $B$) effect as $\M_{0}=\sum_{\x_B=0,1}M_{0 \x_B}$ with convention that we label measurement operators via two-bit strings $\mathbf{\x}=\x_{A}\x_{B}$, where $\x_A,\x_B\in\left\{0,1\right\}$.
We further defined a difference of quantum states as $\Delta = \rho_B-\sigma_B$.
The above optimization is thus performed over 
variable $ \Delta$ that is a traceless operator with condition $-\iden\leq \Delta\leq \iden$ (this follows from the fact that it is a difference of two quantum states). 

Now we propose to parameterize $\Delta$ and treat those parameters as variables in constrained black-box optimization, where cost function is the above operator norm.
Since we use black-box optimization, the method is not exact.
However, the problem involves only single-qubit system and each cost function evaluation is performed very fast, one can thus perform an exhaustive search over parameters space to obtain a good solution.

For AC-distance \cite{MaciejewskiAverageCaseDistance}, we propose to adopt the same strategy, namely optimize over $\Delta$, the following function
\begin{align}
     \sup_{\Delta} \frac{1}{2}\ \sqrt{||\tr_{B}{\left(M_0\ \iden\otimes \Delta\right)}||_{HS}^2+\left(\tr_{B}{\left(M_0\ \iden\otimes \Delta\right)}\right)^2} \ ,
\end{align}
via black-box optimizers.


\section{Channel reductions}
\label{sec:channel_reductions}

The purpose of this section is to show that the reduced POVMs operators introduced in the main section emerge as a special case out of a more general formalism of reduced channels. We will firstly define reduced channels, and then show that a reduction of a quantum-classical channel corresponding to a measurement process yields reduced POVMs from the main text. 

\begin{definition}
Let $\Lambda(\cdot)$ be a $N$-qubit quantum channel. Denote by $A$ a $k$ qubit subset and by $\bar{A}$ its complement. The reduction of the channel $\Lambda (\cdot)$ over $\bar{A}$ is defined as    

\begin{equation}
 \label{eq:reduced_noise_channel} 
     \Lambda^{A}_{\sigma_{\bar{A}}} \left( \rho_{A} \right) \equiv \tr_{\bar{A}} \left( \Lambda \left( \rho_{A} \otimes \sigma_{\bar{A}} \right) \right),
\end{equation}
where $\rho_{A}$ is an input state to the reduced channel, and  $\sigma_{\bar{A}}$ is a fixed state on the complement.
\end{definition}
Note that this definition of reduction is different from the one used in the channel marginal problem \cite{HsiehChannelMarginalProblem2021}, where it is assumed that a valid marginal channel must be independent of the input state to the complement. In the case of cross-talk effects considered here this assumption seems to be too restrictive. 

Let us now consider the reduction of a quantum-classical channel describing the measurement process associated with a POVM $\M$: $\Lambda_M(\rho) = \sum_{\x} \tr \left( M_{\x} \rho \right)  |\x \rangle \langle \x |$. According to Definition 1 we have   
\begin{equation}
\Lambda^{A}_{\M,\sigma_{\bar{A}}} = \sum_{\x_A \x_{\bar{A}}} \text{tr} \left( M_{\x_A \x_{\bar{A}}} \left( \rho_{A} \otimes \sigma_{\bar{A}} \right)   \right)   
\text{tr}_{\bar{A}} \left( |\x_A \x_{\bar{A}} \rangle \langle \x_A \x_{\bar{A}}|  \right) = \sum_{\x_A} \text{tr}_{\bar{A}} \left( M^{A,\sigma_A}_{\x_A } \rho_A  \right) |\x_A \rangle \langle \x_A |,  
\end{equation}
where
\begin{eqnarray}
M^{A,\sigma_A}_{\x_A } \equiv \sum_{x_{\bar{A}}} \tr_{\bar{A}} \left(M_{\x_A \x_{\bar{A}}} \left( \rho_A \otimes I \right) \right).
\end{eqnarray}

Therefore the reduced POVM effects of the main text are recovered from the more general approach of reduction of quantum channels.

\section{Details of the clustering algorithms}
\label{sec:det_clust}

\subsection*{Objective function}
 We start with detailed description of the heuristic objective function $\phi$, whose maximisation realizes the clustering. The function $\phi(\mathcal{P})$ is a function of partition $\mathcal{P}=\lbrace \cluster_1, \cluster_2,\ldots, \cluster_n \rbrace $ of the set of $N$ qubits onto disjoint subsets. The function takes the following form
\begin{align}\label{eq:objective_function}
    \phi\rbracket{\mathcal{P}} =\sum_{i=1}^n S_i-   c_{\text{avg}} \sum_{i=1}^n   f_{\text{penalty}}\rbracket{|\cluster_i|}\ .
\end{align}
In what follows we describe in detail symbols appearing in the above expression. First, the term $S_i$ captures the strength of correlations within  cluster $\cluster_{i}$:  
\begin{align}
        \label{eq:overall_correlation_strength}
        S_{i}\coloneqq \sum_{k,l \in \cluster_i, k \neq l } c_{k \rightarrow l}\ , 
        \end{align}
where correlation coefficients $c_{k\rightarrow l}$ are based on Eq. \eqref{eq:correlation_coefficient_general}. Second, $f_{\text{penalty}}\rbracket{|\cluster_i|}$ is the penalty function for the size of the $i$-th cluster 
\begin{align}
\label{eq:penelty_function}
       f_{\text{penalty}}\rbracket{|\cluster_i|} \coloneqq \begin{cases}
        \alpha\, |\mathcal{C}_i|^2 \quad \text{ if }|\cluster_i| \leq |C_{\text{max}}|\\
        \infty \quad          \\
        \end{cases}.
\end{align}
The role of $f_{\text{penalty}}$ is to ensure that sizes of all clusters do not exceed the value $C_{\text{max}}$. The parameter $\alpha \geq 0$ is introduced to tune the weight of the cluster size penalty, in the main text. Lastly,  $c_{\text{avg}}$ is the average "large correlation coefficient'' defined as follows. Let us consider $N$ qubits. Then we expect to have about $N_{\mathrm{large}} \coloneqq N(C_{\text{max}}-1) $   
``large correlation coefficients'', i.e. the correlation coefficients for qubits within clusters\footnote{This number can be obtained counting how many ordered pairs of qubits are there with with large correlations $c_{k\rightarrow l}$ in the situation in which the system splits onto clusters of strongly correlated qubits of size $C_{\text{max}}$}. Let $C_{\mathrm{large}}$ be the set of pairs $(k,l)$ corresponding to $N_{\mathrm{large}}$ largest correlation coefficients $c_{k \rightarrow l}$. We now define $c_{\text{avg}}$ as
\begin{equation}
\label{eq:cavg}
c_{\text{avg}} \coloneqq \frac{1}{N_{\mathrm{large}}} \sum_{(k,l) \in C_{\mathrm{large}}} c_{k\rightarrow l}   .
\end{equation}
By changing the value of the tuning parameter $\alpha$ it is possible to tune the importance of the penalty term. The structure of clusters $\mathcal{P}$ returned by the algorithms. However, test of the algorithm performed on model data sets show that the final cluster structure does not vary significantly with change of $\alpha$. Therefore, we used the algorithm with $\alpha =0$. 

\subsection*{Maximisation of objective function}

To maximize the cost function $\phi(\mathcal{P})$ we employ an the following local updates of the cluster structure:
i) a move of a selected qubit to a different cluster, and i) a swap of a selected pair of qubits between clusters. We check how value of $\phi\rbracket{\mathcal{P}}$ changes under these operations, and accept changes leading to the largest value of $\phi\rbracket{\mathcal{P}}$. 
The input of our algorithm comprises of: 
 number of runs $N_{\mathrm{runs}}$, maximal cluster size $C_{\mathrm{max}}$, tuning parameter $\alpha$ and correlation coefficients coefficients $\lbrace{c_{i \rightarrow j}\rbrace}$.
Single optimization procedure ends when a convergence criterion is met. Since the choices of elementary operations are randomized, different runs may lead to different clusters structures. The final clustering outputted by the algorithm is the best clustering over $N_{\mathrm{run}}$ runs, i.e. the one with highest value o $\phi(\mathcal{P})$. More details on the algorithm are given in the pseudo-code description below. In order to test whether the algorithm does not get stuck on a local minima we tested its randomized version, in which the acceptance of a change leading to the largest value of $\phi\rbracket{\mathcal{P}}$ is done with some probability. Tests performed on a model data showed that both the randomized and deterministic version of the algorithm output correct cluster structure.

\begin{algorithm}[H]
\small
 \caption{Heuristic clustering algorithm \label{alg:hca_v1}$\qquad \ $}
  \RaggedRight\textbf{Input data}:\\
  \RaggedRight $c_{i \rightarrow j}$: correlation coefficients for each directed pair of $N$ qubits\\
  \RaggedRight\textbf{Input parameters}:\\
  \RaggedRight $N_{\mathrm{runs}}$: number of runs\\
  \RaggedRight $C_{\mathrm{max}}$: maximal number of qubits in cluster\\
  \RaggedRight $\alpha$: tuning parameter\\
\begin{enumerate}
    \item Calculate $c_{avg}$ according to Eq. (\ref{eq:cavg}). The quantity $c_{avg}$ together with $C_{\mathrm{max}}$ and $\alpha$ determines the objective function $\phi$ via Eqs.(\ref{eq:objective_function} - \ref{eq:penelty_function}).
  \item Create the correlation list by computing summarized correlation coefficients $c_{ij}=c_{i \rightarrow j}+c_{j \rightarrow i}$.
  \item Create an initial cluster assignment $\mathcal{P}_{init}$. The clusters are formed out of qubit pairs, in decreasing order of $c_{ij}$ coefficients.
  \item Create the global best cluster assignment $\mathcal{P}_{best}$ and initialize it by $\mathcal{P}_{init}$. 
  \item
   \begin{algorithmic}[1]
   \While{the number of runs is smaller than $N_{\mathrm{runs}}$}
  
      \LState Create the best cluster assignment in the run $\mathcal{P}_{run}$ and initialize it by $\mathcal{P}_{init}$. 
      \Repeat
      \LState Create a list of all unordered pairs of qubits $S_{pairs}$
    
     \While{ the list $S_{pairs}$ is non-empty}
    
          \LState Choose a random pair of qubits $i',j'$ from $S_{pairs}$, and remove them from $S_{pairs}$.

        \If{the qubits belong to different clusters} calculate change of the objective function $\phi$ under
              \LState A move of qubit $i'$ to the cluster of qubit $j'$
              \LState A move of qubit $j'$ to the cluster of qubit $i'$
              \LState A swap of qubits $i',j'$ between their clusters
    
       \If{a better cluster assignment has been found} update $\mathcal{P}_{run}$   
       \EndIf
      \EndIf
     
  \EndWhile
     \Until{$\mathcal{P}_{run}$ hasn't been updated}
    \If{$ \phi( \mathcal{P}_{run})>\phi(\mathcal{P}_{best})$ }     update $\mathcal{P}_{best}$
    \EndIf
    
    \EndWhile
    
      \end{algorithmic}

\item \Return {$\mathcal{P}_{best}$}
 
\end{enumerate}
\end{algorithm}

The choice of particular clustering can be also decided by application or benchmark-driven criteria.  More specifically, the benchmark-based method starts with a scan through different values of $\alpha$.
For a fixed $\alpha$, the chosen clustering algorithm is used to obtain the clusters $\mathcal{P}$ and the corresponding cluster noise model is constructed. Cluster noise models constructed in that way are subsequently tested in a benchmark, and $\alpha$ leading to a model with the highest score is chosen.
It may happen that some values of $\alpha$  lead to the same clusters, and hence to the same cluster noise model.

\section{Additional experimental data}
\label{sec:add_experiment_details}
In Fig. \ref{fig:error_prediction} results of energy prediction benchmark for IBM Cusco and Rigetti Aspen-M-3 are presented. For both devices Clusters and Neighbours noise model outperforms Tensor Product Noise model. In both cases median of prediction error is $\sim 80 \%$ smaller for CN noise model than for TPN model.  


\begin{figure*}[t]
   \includegraphics[width=0.48\textwidth]{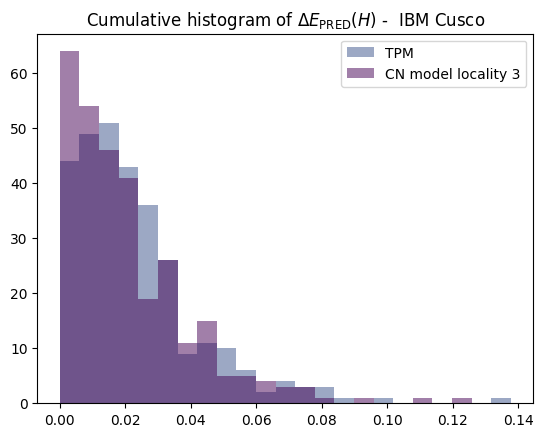} 
    \includegraphics[width=0.48\textwidth]{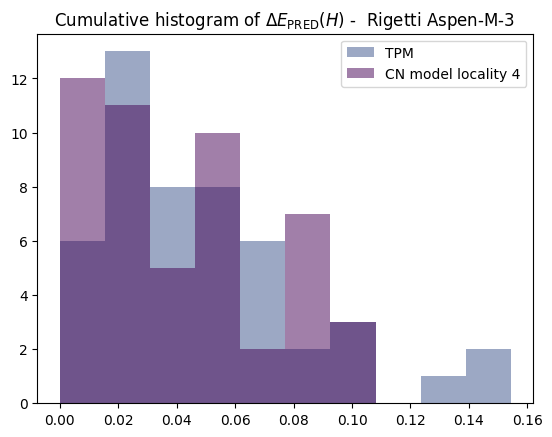}
\caption{\label{fig:error_prediction}
 Results of error-prediction benchmark for \qibm-qubits on IBM Cusco - left panel,
 and \qrig-qubits on Rigetti’s Aspen-M3 - right panel.  
 In both figures distribution of error in energy estimation for error prediction based on product noise model (light purple), correlated clusters noise model (dark purple) is presented. The error is calculated according to Eq. (\ref{eq:benchmark_prediction_error}). 
 The histograms were created from 300 Hamiltonians for IBM Cusco, and 50 Hamiltonians of the same locality Aspen-M3. In both cases Hamiltonians were built from single and two-body terms involving Pauli $Z$ operators. }
\end{figure*}


\end{document}